\documentclass[journal,12pt,onecolumn,draftclsnofoot]{IEEEtran}

\IEEEoverridecommandlockouts
\usepackage{caption}
\usepackage{subcaption}
\usepackage{cite}
\usepackage{bbm}
\usepackage{amsmath}
\usepackage{amssymb}
\usepackage{outlines}
\usepackage{array}
\usepackage{float}
\usepackage{kotex}
\usepackage{graphicx}
\usepackage{bm}
\usepackage{caption,subcaption}
\usepackage[font=footnotesize]{caption}

\usepackage{mathtools}
\usepackage{braket}
\usepackage{bm}
\usepackage{amsthm}
\usepackage{algorithm2e}
\usepackage{algpseudocode}
\RestyleAlgo{ruled}

\usepackage{romannum}

\usepackage{algorithm2e}
\usepackage{algpseudocode}
\usepackage{amsthm}
\usepackage{csquotes}
\usepackage{cite}
\usepackage{multicol}
\usepackage{stmaryrd}
\usepackage{tabularx,booktabs}
\usepackage{xcolor}
\usepackage{url}
\usepackage[
            top=0.75in,
            bottom=1in,
            left=0.625in,
            right=0.625in]{geometry}
\newcolumntype{C}{>{\centering\arraybackslash}X} 
\definecolor{SP}{RGB}{30,160,156}

\usepackage{lipsum} % filler text
\usepackage{amsthm}
\renewcommand{\qedsymbol}{}

\newtheorem{lemma}{Lemma}

\DeclareMathOperator*{\argminB}{argmin}
\def\BibTeX{{\rm B\kern-.05em{\sc i\kern-.025em b}\kern-.08em
    T\kern-.1667em\lower.7ex\hbox{E}\kern-.125emX}}
    
\begin{document}
\pagenumbering{arabic} 
\title{Calibrating Wireless AI via Meta-Learned Context-Dependent Conformal Prediction}
\author{\IEEEauthorblockN{Seonghoon Yoo,~\IEEEmembership{Graduate Student Member,~IEEE,}  Sangwoo Park,~\IEEEmembership{Member,~IEEE,}\\  Petar Popovski,~\IEEEmembership{Fellow,~IEEE,}  Joonhyuk Kang,~\IEEEmembership{Member,~IEEE,}  and\\ Osvaldo Simeone~\IEEEmembership{Fellow,~IEEE}}
\IEEEauthorblockA{ }
% \thanks{S. Yoo, and J. Kang are with the Department of Electrical Engineering, Korea Advanced Institute of Science and Technology, Daejeon 34141, South Korea (e-mail: shyoo902@kaist.ac.kr, jhkang@ee.kaist.ac.kr).}
% \thanks{O. Simeone is with the King’s Communications, Learning
% \& Information Processing (KCLIP) lab within the Centre for Intelligent
% Information Processing Systems (CIIPS), Department of Engineering, King’s
% College London, W2CR 2LS London, U.K. (e-mail: osvaldo.simeone@kcl.ac.uk).}
% \thanks{S. Park is with the Department of Electrical and Electronic Engineering, Imperial College London, London SW7 2AZ, UK. (e-mail:
% s.park@imperial.ac.uk).}
% \thanks{P. Popovski is with the Department of Electronic Systems,
% Aalborg University, Denmark (email: petarp@es.aau.dk).}
% \thanks{This research was supported by the MSIT (Ministry of Science and
% ICT), Korea, under the ITRC (Information Technology Research Center)
% support program (IITP-2025-2020-0-01787) supervised by the IITP (Institute
% of Information \& Communications Technology Planning \& Evaluation)}
}
\maketitle

\begin{abstract}
Modern software-defined networks, such as Open Radio Access Network (O-RAN) systems, rely on artificial intelligence (AI)-powered applications running on controllers interfaced with the radio access network. To ensure that these AI applications operate reliably at runtime, they must be properly calibrated before deployment. A promising and theoretically grounded approach to calibration is conformal prediction (CP), which enhances any AI model by transforming it into a provably reliable set predictor that provides error bars for estimates and decisions. CP requires calibration data that matches the distribution of the environment encountered during runtime. However, in practical scenarios, network controllers often have access only to data collected under different contexts -- such as varying traffic patterns and network conditions -- leading to a mismatch between the calibration and runtime distributions. This paper introduces a novel methodology to address this calibration-test distribution shift. The approach leverages meta-learning to develop a zero-shot estimator of distribution shifts, relying solely on contextual information. The proposed method, called meta-learned context-dependent weighted conformal prediction (ML-WCP), enables effective calibration of AI applications without requiring data from the current context. Additionally, it can incorporate data from multiple contexts to further enhance calibration reliability.
\end{abstract}
\clearpage
\section{Introduction}
\begin{figure*}[h]
    \centering
    \includegraphics[width=0.7\linewidth]{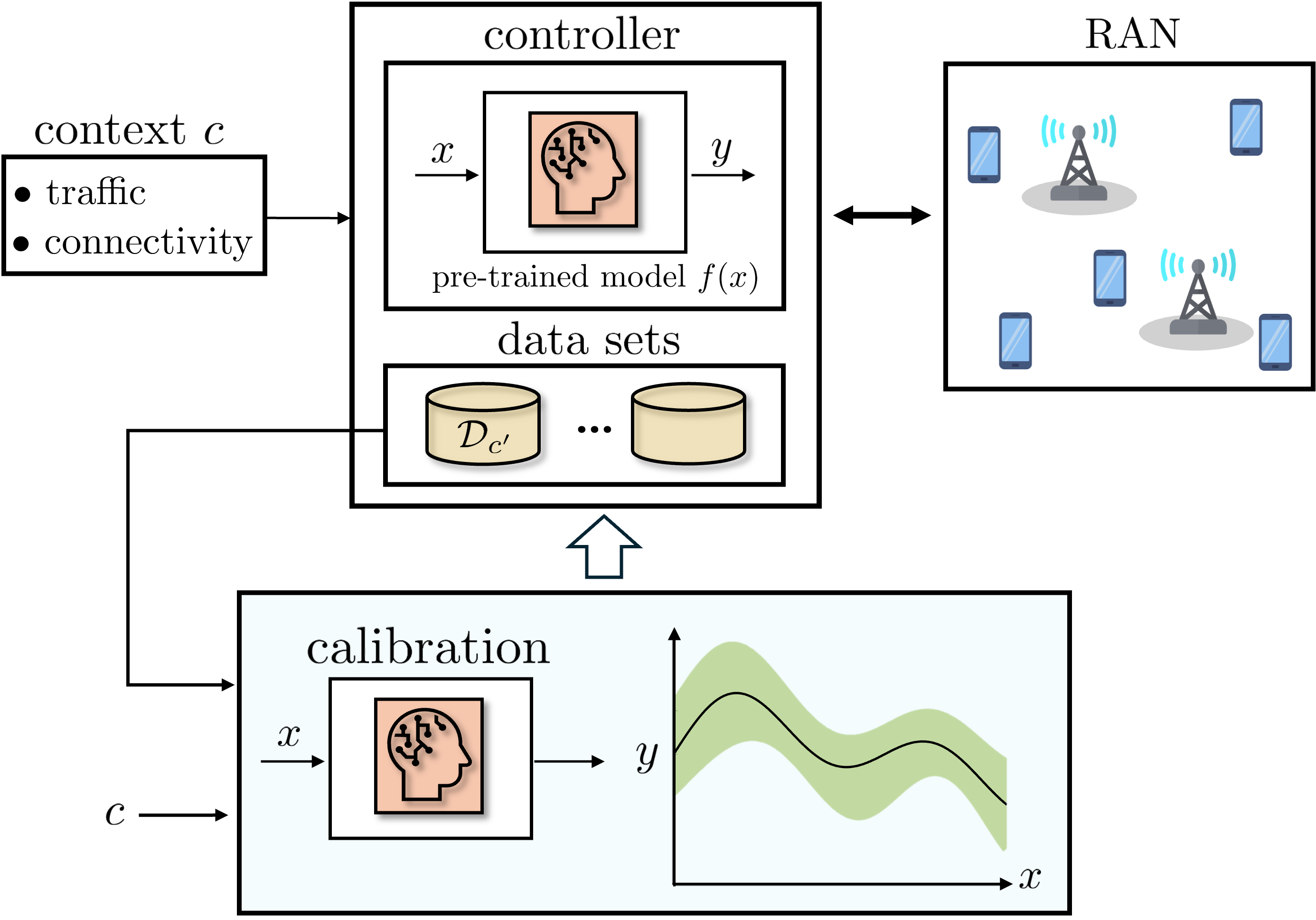}
    \caption{A controller at the cloud or at the edge runs pre-trained AI apps selected based on contextual information $c$ about traffic and connectivity conditions. The app uses data $x$ received from the RAN in order to make decisions $y$ that may affect the operation of the RAN. The goal of this paper is to calibrate the pre-trained AI app prior to deployment with the aim of ensuring reliability. Reliability is obtained by assigning statistically valid error bars -- more generally, prediction sets -- to the decision $y$. Calibration can use data available at the controller that was collected under different contexts $c'$.}
    \label{fig1}
\end{figure*}
\subsection{Context and Motivation}

Modern software-defined networks, such as Open Radio Access Network (O-RAN) systems, rely on artificial intelligence (AI)-powered applications running on controllers interfaced with the radio access network \cite{Abdalla2022oran, Latah2019ai_apps}. As illustrated in Fig. \ref{fig1}, given contextual information encompassing current traffic and connectivity conditions, along with the network operator’s intent, controllers select and deploy pre-trained AI applications. The AI applications acquire data from, and potentially provide feedback to, the radio access network (RAN).

For example, AI in O-RAN can enable functions such as physical layer (PHY) optimization, including adaptive modulation and coding, beamforming, and power control \cite{Bonati2021phy_functions}; as well as higher-layer functions like scheduling, handover management, and resource allocation optimization \cite{qazzaz2024higher_layer}. These functions improve network efficiency, reduce latency, and enhance user experience. Moreover, AI applications can analyze real-time network conditions and dynamically adjust policies, ensuring efficient adaptation to traffic changes and reliable service quality \cite{Brik2024oran_automation, Hamdan2023ORAN_automation}.

To ensure that AI applications operate reliably at runtime, they must be properly calibrated before
deployment. Calibration broadly refers to the selection of hyperparameters or post-processing mechanisms
that provide guarantees of reliability during operation \cite{guo2017calibration, kuleshov2018accurate, marx2022modular, angelopoulos2024therory_CP, shafer2008CP,balasubramanian2014cp_book, YANG2020hyperparam_review}. For instance, for a top-$K$ traffic predictor, it is critical to choose the list size, $K$, that ensures a sufficiently large coverage probability for the traffic class \cite{Angelopoulos2020Top-k}. This way, downstream applications can make a reliable use the prediction, being aware of other plausible future outcomes. 

Given the increased use of AI in sensitive domains, such as telecommunications and engineering, the problem of AI calibration has garnered significant attention in recent machine learning literature \cite{Minderer2021NEURIPS, wang2024ml_calibration}. A flexible and theoretically grounded approach to calibration is \textit{conformal prediction (CP)}. CP augments any AI model by transforming it into a provably reliable set predictor that provides error bars for estimates
and decisions \cite{angelopoulos2024therory_CP}. Specifically, as seen in Fig. \ref{fig1}, CP optimizes a post-processing mechanism to add error bars – more generally, prediction sets – to the outputs of an AI model. The working principle of CP involves estimating the distribution of errors made by the AI model through a held-out \emph{calibration data set} to determine the size of error bars that are likely to include the correct, or optimal, output.

However, in practical scenarios, network controllers typically have access only to data collected under different contexts -- such as varying traffic patterns and network conditions -- that do not necessarily match runtime conditions. This leads to a mismatch between the calibration data distribution and the runtime data distribution, impeding a direct application of CP. 

CP has been recently extended to handle distribution shifts between calibration and testing \cite{angelopoulos2024therory_CP}. The seminal paper \cite{Tibshirani2019Dist_shift} considered \emph{covariate shifts}, i.e., discrepancies between input distributions in the calibration and runtime phases, that are known \emph{a priori}. The authors of \cite{Tibshirani2019Dist_shift} show how to apply importance weighting to ensure the statistical validity of the prediction set despite the covariate shift. The resulting scheme is known as \emph{weighted conformal prediction} (WCP). Several works have also studied the case of an \emph{unknown covariate shift} \cite{Lei2021TV, Yang2024DoublyRobustCP}. More general distribution shifts, encompassing also concept shifts have been considered in \cite{prinster2024conformal, barber2023conformal, cauchois2024robust}. While guaranteeing reliable error bars, WCP can yield large error bars, unless the underlying model and error (score) functions are suitably designed  \cite{Lei2021TV, hou2024likelihood}.

%Existing methods, such as \textit{weighted conformal prediction (WCP)} \cite{Tibshirani2019Dist_shift,Lei2021TV, Yang2024DoublyRobustCP}, attempt to address this issue by reserving data from both runtime and calibration distributions to estimate the distribution shift, unless this shift is known a priori. However, these approaches may result in inefficient use of limited data.

In this work, we study the problem of calibration of AI apps for wireless systems. Unlike existing studies,
in a wireless system, distribution shifts are typically unknown, and there are no additional data available
for the current context to estimate the distribution shift. For example, the traffic statistics may change over time, making data collected under past settings not directly representative of current traffic conditions. 

To address this calibration in the presence of unknown distribution shifts, we introduce a novel
methodology that leverages meta-learning \cite{Sangwoo2023meta-learning} to develop a zero-shot estimator of distribution shifts, relying
solely on contextual information. Contextual variables describe the network operational
conditions in terms of traffic – e.g., number of users, the types of services being delivered, the cell average
loads, and mobility levels –, as well as of connectivity – e.g., topology, fronthaul capacities, average
signal-to-noise ratio (SNR) conditions, and multi-modal sensory data from cameras, GPS, or Radar. While current labeled data may not be available, contextual information is typically inherently accessible by a network controller \cite{Wang2024context_controller, Semiari2015Context}. The proposed method, called \textit{meta-learned context-dependent weighted conformal prediction (ML-WCP)}, enables the effective calibration of AI applications without requiring data from the current context. Furthermore, it can incorporate data from multiple contexts to enhance calibration reliability.

\subsection{Related Works}
Here, we briefly review additional relevant papers related to
the theme of this work.

\textit{Reliable uncertainty quantification for AI in wireless systems}: While AI has shown great potential to enhance communication systems \cite{Letaief2019AI,Zappone2019AI}, its practical deployment has been hindered by concerns on reliability and verification \cite{guo2020explainable}. Bayesian learning \cite{simeone2022machine} provides a principled framework to quantify the uncertainty of AI models, and it has been successfully applied to wireless systems \cite{cohen2021learning, raviv2023modular, tedeschini2024real}. However, Bayesian learning is limited by its sensitivity to prior and model misspecifications and by the high computational complexity \cite{knoblauch2019generalized, Zecchin2023Bayesian}. Furthermore, Bayesian learning is not applicable to the scenarios in which only a pre-trained AI model is available to the system. In contrast, post-hoc calibration methods \cite{guo2017calibration} can reliably quantify the uncertainty of pre-trained AI models using a held-out data set. Notably, CP \cite{vovk2005cp} yields provably valid prediction sets around the decisions made by AI models. 

\textit{CP in wireless systems:} 
The application of CP to wireless systems has been demonstrated for functionalities including demodulation \cite{Cohen2023WirelessCP}, network quality estimation \cite{jiang2024learning}, and traffic prediction \cite{ma2023metastnet}. WCP has been leveraged to address the counterfactural estimation of key performance indicators for wireless AI apps \cite{hou2024likelihood}. Online versions of CP \cite{gibbs2021adaptive, feldman2022achieving} have been considered for settings with feedback, including scheduling \cite{Cohen2023URLLC_CP}.

\subsection{Main Contributions}
% In this work, we address this problem by leveraging the availability of rich contextual information to optimize zero-shot estimators of the covariate likelihood ratio that depend solely on the test context $c^\text{te}$ and the calibration context $c^\text{cal} \in \mathcal{C}^\text{cal}$. As shown in Fig. \ref{fig1}, the zero-shot estimator is meta-trained by using data from multiple contexts, which is readily available in many engineering domains such as addressing burst interference issues in physical-layer applications or traffic prediction in 5G wireless open radio access network (O-RAN) \cite{Chowdbury[2024]TrafficPred.}. Based on pre-trained models, we propose an ML-WCP framework for generating prediction sets and validating their performance through simulation results. Our main contributions are as follows.
This paper addresses the challenge of calibrating AI applications in wireless systems by leveraging contextual information, making the following main contributions:
\begin{itemize}
    \item \textit{Novel context-based zero-shot calibration methodology}: This paper introduces ML-WCP, a methodology for estimating calibration-test distribution shifts without requiring runtime data. This scheme enables zero-shot calibration using only contextual information, and is applicable to any setting in which calibration data are available from multiple contexts. ML-WCP builds on meta-learning and on efficient symmetry-based neural model architectures, and is capable of integrating data from multiple contexts.
   \item \textit{Applications to wireless systems}: We demonstrate three applications of ML-WCP across different layers of a communication network: traffic slice prediction at the network layer, scheduling apps profiling at the medium access control (MAC) layer, and interference-limited communication at the physical layer.
    \item \textit{Experimental validation}: The performance of ML-WCP is validated through extensive experiments, supporting the theoretical results that coverage performance depends on the quality of the covariate likelihood estimator.  %[please expand] 아마 여기서 coverage performance는 estimator omega의 quality에 따라 결정되며 이에 대한 theoretical result를 실험에서 검증했다는 걸 설명해야 함.
\end{itemize}

\subsection{Organization}
The remainder of the paper is organized as follows. Section \ref{Problem Definition} defines the problem, covering context-dependent data generation and examples for wireless applications. Section \ref{Background} summarizes the necessary background material on CP and WCP. Section \ref{Meta-Learned Context-Dependent Weighted Conformal Prediction} presents the proposed ML-WCP, and Section \ref{ML-WCP with Multi-Context Calibration} extends ML-WCP to multi-context calibration. The experimental setting and results are described in Section \ref{Simulation results}. Finally, Section \ref{Conclusion} summarizes the paper and points to directions for future work.

\section{Problem Definition} \label{Problem Definition}
\begin{figure*}[t]
    \centering
    \includegraphics[width=0.9\linewidth]{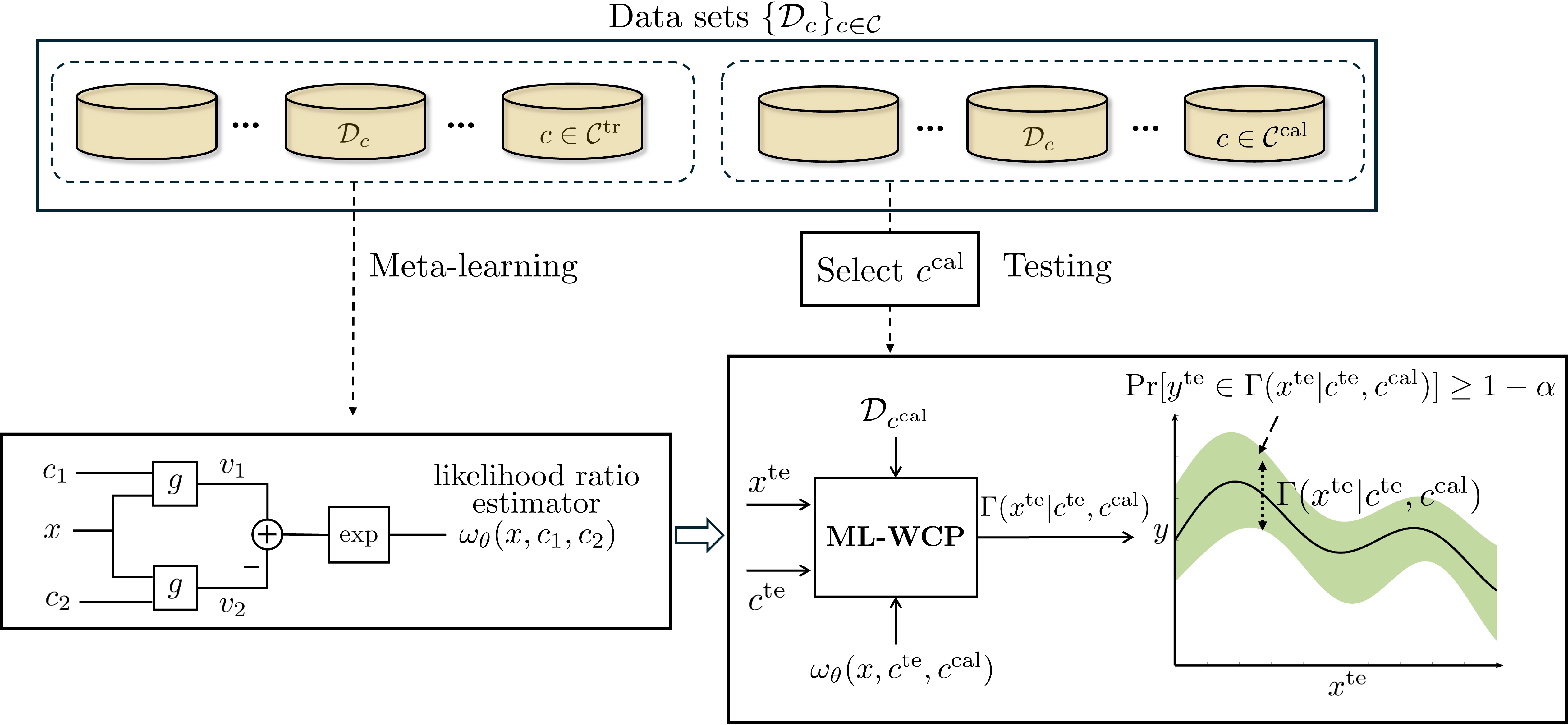}
    \caption{(a) Meta-learning: Meta-learned context-dependent weighted conformal prediction (ML-WCP) leverages calibration data from multiple contexts to meta-learn a zero-shot covariate likelihood ratio estimator $\omega_{\theta}(x,c_1,c_2)$. (b) Testing: Given a selected calibration set $c^\text{cal}$, ML-WCP provides a set predictor $\Gamma(x^\text{te}|c^\text{te},c^\text{cal})$ that aims at guaranteeing coverage for an input $x^\text{te}$ from a new context $c^\text{te}$.}
    \label{fig2}
\end{figure*}
This paper considers the scenario in Fig. \ref{fig1}, in which a controller runs a pre-trained app at the cloud or at the edge. The controller has access to a number of data sets collected during the past operation of the system. The goal is to leverage the available data in order to calibrate the pre-trained AI model via a low-complexity post-processing mechanism operating on the model's outputs. This section presents the problem under study, while also illustrating a number of example applications to wireless networks, which will be further elaborated on in Section \ref{ex_wireless communication}.

\subsection{Calibrating a Pre-Trained App}
The controller selects a  \textit{pre-trained model} $f(x)$ that associates a decision $y \in \mathcal{Y}$ to any input $x \in \mathcal{X}$. We will refer to the decision $y$ also as a \emph{prediction}, although function $f(x)$ may carry out other forms of decision making beyond prediction. We are interested in designing low-complexity \emph{post-hoc calibration} strategies that leverage the model's output, $f(x)$, to produce decisions associated with a \emph{reliability guarantee}. Specifically, following the CP framework \cite{angelopoulos2024therory_CP,Anastasios2023CPintro,vovk2005cp,shafer2008CP}, we wish to produce \emph{prediction set} that are guaranteed to include the true, or optimal, output $y$ with a user-defined probability.

% for each new test input $x^\text{te}$, a $\Gamma(x^\text{te})\subseteq y$ that includes the true output $y^\text{te}$ with a user-defined target probability.

% We consider a supervised learning setting, with input variables denoted as  and target variables as .

To facilitate calibration, we assume that, at runtime, side information about the current network conditions is available in the form of \emph{contextual variables} $c$. The context variable $c$ is a vector of features taking values in some set $\mathcal{C}$. Importantly, the available pre-trained predictive model $f(x)$ is not specialized to context $c$, as this information is only available at runtime, but this information can be leveraged to post-process the output model $f(x)$ with the goal of enhancing calibration.  %Given the data sets $\{\mathcal{D}_{c}\}_{c \in \mathcal{C}^\text{tr}}$ and the new context-input test pair $(c^\text{te},x^\text{te})$, the goal is to generate a prediction set $\Gamma(x^\text{te})$ with the reliability guarantees to be specified below.

 Given a test input $x^\text{te}$ and the corresponding test context $c^\text{te}$, the goal of calibration is to use the pre-trained predictor $f(x)$ to evaluate a prediction set $\Gamma(x^\text{te}) \subseteq \mathcal{Y}$ with the property that the true optimal output  $y^\text{te}$ is included in the set with probability at least $1-\alpha$, i.e.,
\begin{equation}
    \text{Pr}[y^\text{te}\in \Gamma(x^\text{te}|c^\text{te})] \geq 1-\alpha,
    \label{set_predictor}
\end{equation}
where $\alpha$ is the user-defined miscoverage level. For example, the set $\Gamma(x^\text{te}|c^\text{te})$ may take the form of error bars, or confidence intervals, when $y$ is a continuous quantity, while corresponding to a subset of plausible options in the case of a discrete output set $\mathcal{Y}$.

Since the condition (\ref{set_predictor}) can be always guaranteed by choosing the uninformative prediction set $\Gamma(x^\text{te})=\mathcal{Y}$, it is necessary to evaluate the performance of the prediction set in terms of its average size, referred to as \emph{inefficiency}, which is defined as the expectation 
\begin{equation}
  \textnormal{inefficiency}(\Gamma)=\mathbb{E}\big[|\Gamma(x^\text{te}|c^\text{te})|\big].
  \label{inefficiency}
\end{equation}
To this end, calibration aims at ensuring condition (\ref{set_predictor}), while keeping the inefficiency as low as possible.

As discussed later in this section, the set $\Gamma(x^\text{te}|c^\text{te})$ can be constructed by using data available at the controller that was collected under contexts, i.e., network and connectivity conditions, that are different from the one represented by the current context vector $c^\text{te}$.

\subsection{Data Distribution}

Given a context $c \in \mathcal{C}$, the distribution of an input-output data pair  $({x},y)$ is denoted as $p({x},y|c)$. Specifically, we assume a \emph{covariate-shift} setting in which the context-conditional distribution factorizes as 
\begin{equation}
    p(x,y|c)=p(x|c)p(y|x).
    \label{pb_condition}
\end{equation} 
According to the equality (\ref{pb_condition}), the conditional distribution $p(y|x)$ remains constant across contexts, while the distribution of the input, or covariate $x$, $p(x|c)$, determines variations in the joint distribution 
$p(x,y|c)$ as a function of the context $c$. The assumption (\ref{pb_condition}) implies that there exists an ideal  model $p(y|x)$ that describes the relationship between input and output across all contexts $c$, although the likelihood of different inputs $x$, modeled by $p(x|c)$, varies as a function of the context variables $c$. 

The covariate-shift working assumption (\ref{pb_condition}) is aligned with the use by the controller of a pre-trained model $f(x)$, which does not depend on the context $c$. In fact, under this assumption, it is reasonable to adopt a single model $f(x)$ for different contexts. However, this assumption also implies that the model $f(x)$ is not to tailored to the context-specific operating conditions determined by the input distribution $p(x|c)$. Specifically, the model $f(x)$ may be more or less effective depending on the current context-specific distribution $p(x|c)$, which determines which inputs are more likely to be observed. For instance, a traffic classification model trained with data from high mobility users may have suboptimal performance in settings with static users. The role of the calibration step is to account for the errors made by the model as a function of the context $c$ so as to identify reliable prediction sets.

\subsection{Calibration Data}
\label{calibration data}
As illustrated in Fig. \ref{fig1}, in order to determine the prediction set $\Gamma(x^\text{te}|c^\text{te})$,  we assume that the controller can leverage data sets of the form
\begin{equation}
    \mathcal{D}_c=\{(x[i],y[i])\}_{i=1}^{|\mathcal{D}_c|},
    \label{dataset}
\end{equation}
which are collected under contexts $c$ belonging to some set  $\mathcal{C}$. Accordingly, each $i$-th sample consists of an input-output pair $(x[i],y[i])$ that follows the conditional distribution $p(x,y|c)$.

The runtime, or test, context $c^\text{te}$ is generally not included in the set $\mathcal{C}$ of contexts for which data are available at runtime for calibration. Therefore, by the assumption (\ref{pb_condition}), the test data distribution $p(x,y|c^\text{te})$ is different from all the distributions $p(x,y|c)$, with $c\in \mathcal{C}$, for which data are available. This creates the challenge of accounting for the covariate shift between calibration and testing data distributions.

At runtime, given the test context $c^{\mathrm{te}}$, the controller selects some data sets $\mathcal{D}_{c'}$ for given contexts $c' \in \mathcal{C}$. The selection of the contexts $c' \in\mathcal{C}$ can be done using an arbitrary criterion involving only the context variables. We start by considering the selection of a single calibration context $c^\text{cal}\in \mathcal{C}$, and we study the extension to multiple calibration contexts in Section \ref{ML-WCP with Multi-Context Calibration}.

The calibration context $c^\text{cal}$, along with the corresponding calibration data set $\mathcal{D}_{c^\text{cal}}$, can be selected based on a given distance measure $d(\cdot,\cdot)$ in the space of contexts. In practice, the controller can choose the context $c^\text{cal}$ as the context in set $\mathcal{C}$, or a subset thereof (see Section \ref{Meta-Learned Context-Dependent Weighted Conformal Prediction}), at the minimum distance from $c^\text{te}$, i.e.,
\begin{equation}
    c^\text{cal} = \argminB_{c \in \mathcal{C}} d(c^\text{te}, c).
    \label{single_cal_select}
\end{equation}

%Given test context $c^\text{te}$ and the selected calibration context $c^\text{cal}$ sufficiently satisfies the conditions specified in Lemma 1 and (\ref{omega}). Consequently, WCP can be applied with $\omega_{\theta}(x,c^\text{te},c^\text{cal})$ in lieu of $w(x,c^\text{te},c^\text{cal})$ for the selected contexts.

The data set $\mathcal{D}_{c^\text{cal}}$ is used to determine the prediction set $\Gamma(x|c^\text{te})$ through a low-complexity mechanism to be discussed in the next section. Accordingly, we adopt the more detailed notation $\Gamma(x|c^\text{te},c^\text{cal})$. Given the dependence of the prediction set on contexts $c^\text{te}$ and $c^\text{cal}$, the probability in (\ref{set_predictor}) and the expectation in (\ref{inefficiency}) are evaluated with respect to the joint distribution of test and calibration data given the respective context vectors, i.e., 
\begin{equation}
p(\mathcal{D}_{c^\text{cal}},x^\text{te},y^\text{te}|c^\text{te},c^\text{cal})=p(x^\text{te},y^\text{te}|c^\text{te})\cdot p(\mathcal{D}_{c^\text{cal}}|c^\text{cal}),
\label{joint_dist_eq}
\end{equation}
with
\begin{equation}
    p(\mathcal{D}_{c^\text{cal}}|c^\text{cal})=\prod_{i=1}^{|\mathcal{D}_{c^\text{cal}}|}p(x[i],y[i]|c^\text{cal})
    \label{c_cal_eq}
\end{equation} accounting for the standard assumption of independent identically distributed (i.i.d.) data.

%In this section, we define a set predictor and introduce the performance metrics, coverage, and size. We then present the conventional hard-decision-based set generation method, followed by the CP-based set predictor, and extend it to a weighted CP-based predictor under covariate shift.

%We consider methods that select a single calibration context $c^\text{cal} \in \mathcal{C}^\text{cal}$ in order to construct the prediction set $\Gamma(x^\text{te}|c^\text{te},c^\text{cal})$, which is thus a function also of the data set $\mathcal{D}_{c^\text{cal}}$. The calibration context $c^\text{cal}$ can be selected according to any criterion, as long as $\Gamma(x^\text{te}|c^\text{te},c^\text{cal})$ is based solely on the context vector $c^\text{te}$ and $c^\text{cal}$. The distribution under which the probability (\ref{set_predictor}) is calculated is thus given by the product $p(x^\text{te},y^\text{te}|c^\text{te})\cdot p(\mathcal{D}_{c^\text{cal}}|c^\text{cal})$, with $p(\mathcal{D}_{c^\text{cal}}|c^\text{cal})=\prod_{i=1}^{|\mathcal{D}_{c^\text{cal}}|}p(x[i],y[i]|c^\text{cal})$, where each input-output pair $(x[i], y[i])$ is sampled from the set $\mathcal{D}_{c^\text{cal}}$.

\subsection{Examples}
\label{ex_wireless communication}
\begin{figure}[t]
     \centering
     \begin{subfigure}[t]{0.7\textwidth}
         \centering
         \includegraphics[width=\textwidth]{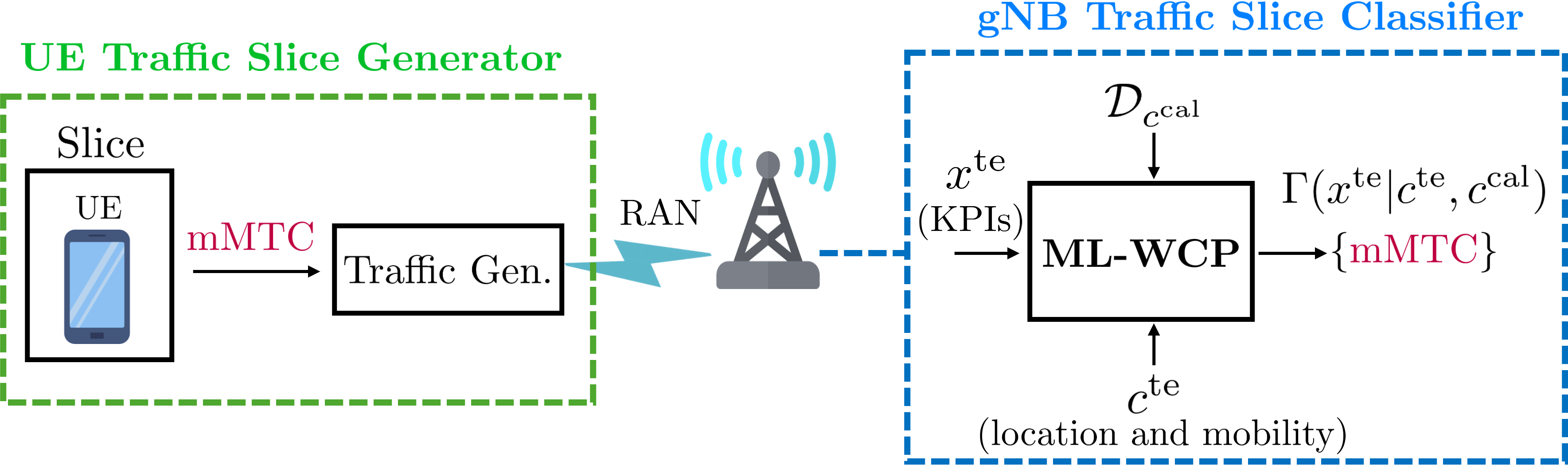}
         \caption{}
         \label{app1}
     \end{subfigure}\\
     \begin{subfigure}[t]{0.75\textwidth}
         \centering
         \includegraphics[width=\textwidth]{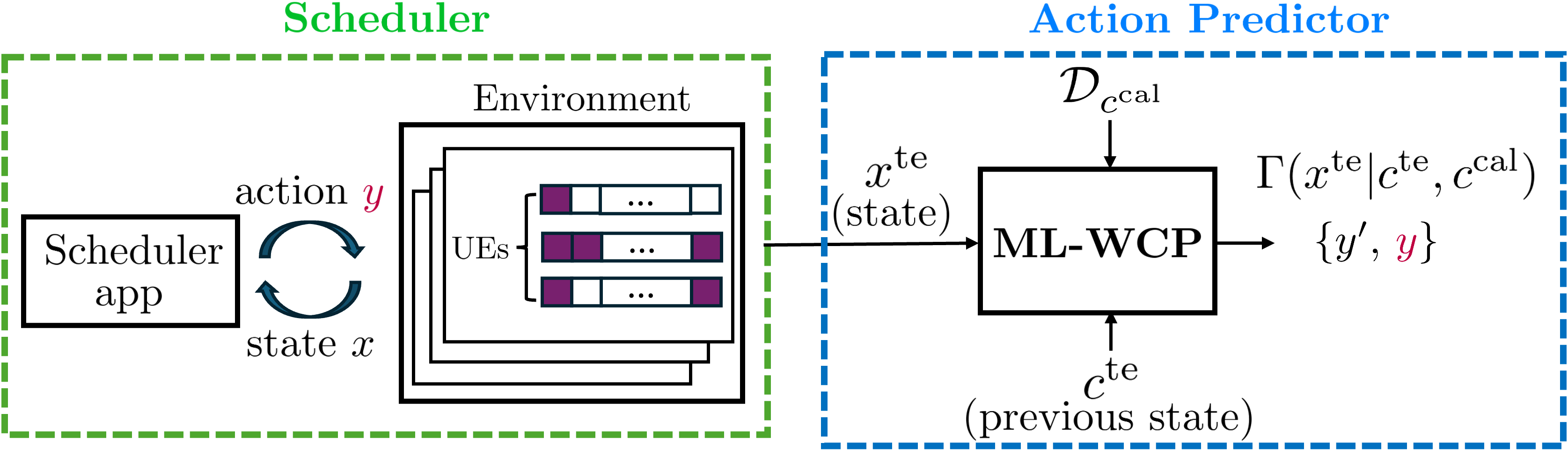}
        \caption{}
         \label{app2}
     \end{subfigure}\\
     \begin{subfigure}[t]{0.75\textwidth}
         \centering
         \includegraphics[width=\textwidth]{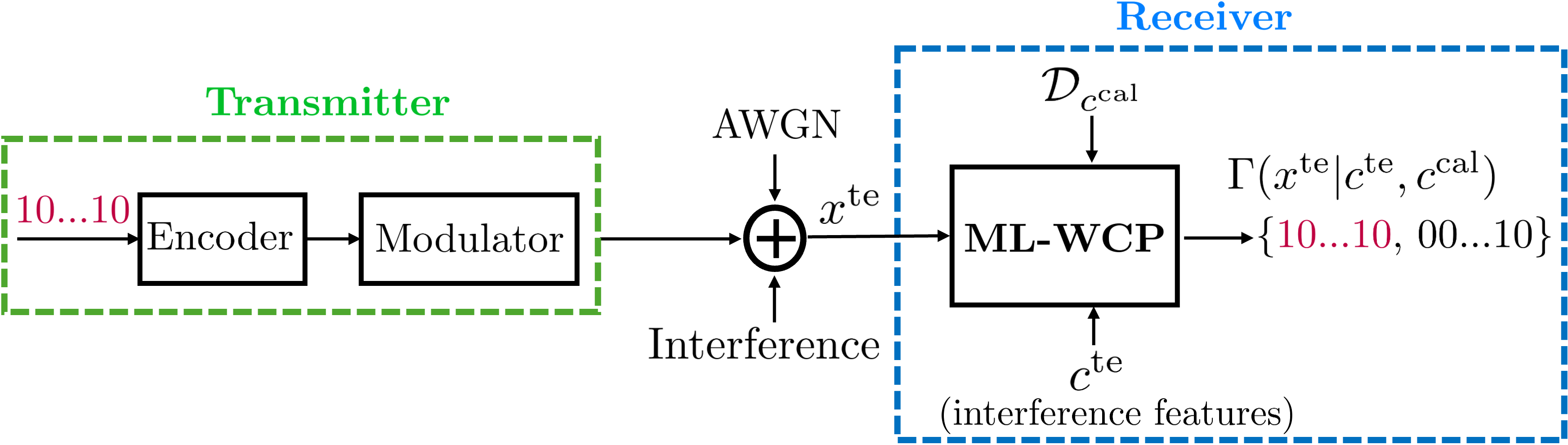}
        \caption{}
         \label{app3}
     \end{subfigure}
        \caption{Examples of applications of ML-WCP to wireless systems: (a) traffic slice prediction, (b) profiling medium access control scheduling apps, and (c) interference-limited physical-layer communication.}
        \label{application}
\end{figure}

  To illustrate the calibration framework studied in this work, we now present some examples of applications at different layers of a communication network that will be further elaborated on in Section \ref{Simulation results}.
\subsubsection{Traffic slice prediction}
\label{ex_traffic prediction}
For a network-layer application, consider a traffic slice classification task in an O-RAN system \cite{Chowdbury[2024]TrafficPred., Petar2023ORAN}. As illustrated in Fig. \ref{application}(a), in this problem, the goal is to estimate the type of traffic produced by a user equipment (UE), which may be enhanced mobile broadband (eMBB), massive machine-type communication (mMTC), ultra-reliable low-latency communication (URLLC), or control. The input covariates used by the pre-trained predictor encompass a number of measured key performance indicators (KPIs). The context $c$ refers to additional information available at the controller at runtime about the location and mobility of the UE whose traffic is being classified. For example, using the format of the data set provided by \cite{Lei2021TV}, the context $c$ may encompass the UE's general location, i.e., campus or residential, as well as the mobility status of the UE, i.e., stationary, walking, or driving. In this case, the distribution $p(x|c)$ describes the properties of the KPIs for UEs with given location and mobility levels. Furthermore, the invariant distribution $p(y|x)$ describes the relationship between the classification decision $y$ and the observed KPIs $x$. The goal is to calibrate the classifier, so that it produces a list of possible traffic types that includes the true type with high probability.
% basdad

\subsubsection{Profiling medium access control scheduling apps}
\label{ex_mac scheduling}
In an open architecture like O-RAN, it is often important to profile the third-party apps in order to understand and evaluate their operations \cite{Hoffman2024third_party, santos2024third_party}. For example, consider scheduling an AI app for the MAC layer. As illustrated in Fig. \ref{application}(b),  the scheduling app selects a subset of UEs for transmission based on their KPI requirements, backlogs, and channel quality indicators (CQIs). To profile the app, the controller wishes to calibrate a predictive model for the decision made by the scheduler. The predictive model can be useful for diagnostic and monitoring purposes \cite{Brik2024oran_automation}. 

At runtime, the controller has access to context variables $c$ that provide information about the backlogs and CQIs, at the previous scheduling interval. This information yields the distribution $p(x|c)$ over the current backlog and CQI $x$. The invariant distribution $p(y|x)$ describes the action $y$ of the app given the current setting $x$.   % Therefore, the context $c$ is composed of four hyperparameters: $c = \{c_1, c_2, c_3, c_4\}$, where each parameter corresponds to the channel quality for each UE, the total queue sizes across UEs, the age of the oldest packet in each UE’s buffer, and the fairness measure related to the fraction of resource blocks that were previously allocated to each UE.

\subsubsection{Interference-limited physical-layer communication}
\label{ex_interference_limited_communication}
Consider an interference-limited communication scenario in which interference occurs in bursts \cite{Fertonani2009Inf}. As illustrated in Fig. \ref{application}(c), a decoding app $f(x)$ is selected to map the received baseband signal $x$, a complex vector, into an estimate of the sequence of information bits $y$.  The decoding app is designed for an additive-Gaussian noise channel using a minimum-distance decoding metric, and is not tailored to specific interference conditions. Thanks to calibration, we wish to produce a \emph{list} of possible information messages that is likely to include the true message \cite{guruswami2004list, Viterbo2021ListViterbi}.

%Accordingly, the input and output domains are given by $\mathcal{X}=\mathbb{C}^N$, where $N$ is the number of encoded symbols, and the output space is $\mathcal{Y}=\{0,1\}^K$, where $K$ is the number of information bits, as shown in Fig. \ref{application}(a). A pre-trained predictor corresponds $f(x)$ to a specific decoding algorithm.

Information about the interference may become available at runtime, e.g., via the use of sensors or spectrum detection algorithms \cite{NING2015Spectrum, Semiari2015Context}. These contextual variables may accordingly describe characteristics such as the probability of an interference burst within a frame and the expected duration of an interference burst. The context vector $c$ can be more or less informative depending on its level of specificity. For instance, a context vector $c$ may include only the probability of occurrence of an interference burst in a frame and the duration of the burst; while a more informative context vector $c$ may include a binary indicator for whether interference occurs, or not, along with the exact duration of the interference burst. This information can be used to calibrate the post-processing block that produces a list of messages based on the decoder's output. 

In this setting, the context-specific distribution $p(x|c)$ accounts for the dependence of the received signal $x$ on the properties of the interference $c$. Furthermore, the invariant distribution $p(y|x)$ describes the operation of a universal decoder, operating without knowledge of the context $c$. It is noted that in this example, the working assumption of a context-invariant distribution $p(y|x)$ does not hold, since the optimal decoder generally depends on $c$. Nevertheless, despite this modeling mismatch, the results in Section \ref{Interference-Limited Communication} will demonstrate that the proposed approach is still effective at calibrating a pre-trained decoder model $f(x)$. Thus, this setting provides a benchmark to evaluate the impact of a modeling mismatch.

\section{Background} \label{Background}
In this section, we review conventional calibration via CP \cite{vovk2005cp}, as well as via its extension WCP \cite{Tibshirani2019Dist_shift} which can operate under covariate shifts.
\subsection{Conformal Prediction}
\label{Conformal Prediction}
Using the output of the pre-trained model $f(x^\text{te})$ for the test input $x^\text{te}$, CP evaluates a negatively oriented score $S(x^\text{te},y)$ -- the smaller the better -- for all the possible outputs $y\in\mathcal{Y}$. The prediction set is obtained by including all the outputs $y$ whose score is below some threshold $\gamma $, i.e., 
\begin{equation}
\Gamma(x^\text{te}|c^\text{te},c^\text{cal})=\bigg\{y \in \mathcal{Y}: S(x^\text{te},y) \leq \gamma \bigg\}. 
\label{gamma}
 \end{equation} Note that this is a low-complexity thresholding mechanism, requiring only the evaluation of the output of the pre-trained model. Depending on the application, the score $S(x,y)$ may be selected as the residual score $S(x,y)=|y-f(x)|$ if $f(x)$ is a regression model, or as the high-probability score $S(x,y)=f_y(x)$ if $f(x)$ is a classification model assigning probability $f_y(x)$ to class $y$. Other examples of scores are reviewed in \cite{angelopoulos2024therory_CP, NEURIPS2019Conformalized}. 
 
Given a target miscoverage rate $\alpha$, CP computes the threshold  $\gamma$ in (\ref{gamma}) based on the calibration data set $\mathcal{D}_{c^\text{cal}}$. Specifically, it obtains the prediction set as 
\begin{equation}
\Gamma_{\alpha}^{\text{CP}}(x^\text{te}|c^\text{te},c^\text{cal})=\bigg\{y \in \mathcal{Y}: S(x^\text{te},y) \leq \textrm{Quantile}\bigg({1-\alpha};\sum_{i=1}^{|\mathcal{D}_{c^\text{te}}|}\delta_{S(x[i],y[i])}+\delta_{\infty}\bigg)\bigg\},    
\label{cp}
    \end{equation}
where $\delta_{(\cdot)}$ is the Dirac delta function, and $\textrm{Quantile}\big({\beta};p\big)$ represents the $\beta$-quantile of the empirical distribution $p$. Accordingly,  the right-hand side of (\ref{cp}) can be equivalently computed as the $\lceil (1-\alpha)(|\mathcal{D}_{c^\text{te}}|+1) \rceil$-th smallest value in the set $\{S(x[1],y[1]),...,S(x[|\mathcal{D}_{c^\text{te}}|
],y[|\mathcal{D}_{c^\text{te}}|])\}$ \cite{angelopoulos2024therory_CP}.

Under the assumption that calibration and test contexts are identical, the set (\ref{cp}) satisfies the coverage requirement (\ref{set_predictor}), i.e., \cite{vovk2005cp,Anastasios2023CPintro}
\begin{equation}
\text{Pr}\big[y^\text{te} \in \Gamma_{\alpha}^\text{CP}(x^\text{te}|c^\text{te})\big]\geq 1-\alpha,    
\end{equation}
where the probability is computed over the distribution (\ref{joint_dist_eq})-(\ref{c_cal_eq}) with $c^\text{te}=c^\text{cal}$. 

In contrast, if the calibration data set $\mathcal{D}_{c^\text{cal}}$ corresponds to a different context $c^\text{cal} \neq c^\text{te}$, a direct application of CP as in (\ref{cp}) would yield the prediction set
\begin{equation}
\Gamma_{\alpha}^{\text{CP}}(x^\text{te}|c^\text{te},c^\text{cal})=\bigg\{y \in \mathcal{Y}: S(x^\text{te},y) \leq \textrm{Quantile}\bigg({1-\alpha};\sum_{i=1}^{|\mathcal{D}_{c^\text{cal}}|}\delta_{S(x[i],y[i])}+\delta_{\infty}\bigg)\bigg\}.    
\label{cp_shift}
    \end{equation}
The prediction set has the weaker coverage guarantee \cite{barber2023exc}
\begin{equation}
\text{Pr}\big[y^\text{te} \in \Gamma_{\alpha}^\text{CP}(x^\text{te}|c^\text{te},c^\text{cal})\big]\geq 1-\alpha-\lVert p(x|c^\text{cal})-p(x|c^\text{te})\rVert_\text{TV},
\label{cp_coverage}
    \end{equation}
where $\lVert \cdot \rVert_\text{TV}$ is the total variation (TV) between the input distributions $p(x|c^\text{te})$ and $p(x|c^\text{cal})$, i.e.,
\begin{equation}
\lVert p(x|c^\text{cal})-p(x|c^\text{te})\rVert_\text{TV}=\frac{1}{2}\int| p(x|c^\text{cal})-p(x|c^\text{te})|\mathrm{d}x.
\label{tv}
\end{equation}
By (\ref{cp_coverage}), if the contexts are different, $c^\text{te} \neq c^\text{cal}$, and thus $\lVert p(x|c^\text{cal})-p(x|c^\text{te})\rVert_\text{TV}>0$, the predicted set no longer satisfies the desired target coverage $1-\alpha$. Furthermore, the corresponding coverage gap can be bounded by the distance $\lVert p(x|c^\text{cal})-p(x|c^\text{te})\rVert_\text{TV}$ between the covariate distributions under different contexts $c^\text{cal}$ and $c^\text{te}$  \cite{Lei2021TV, barber2023exc}. 

\subsection{Conservative Conformal Prediction}
\label{conservative conformal prediction}
Based on the result (\ref{tv}), if an estimate of the TV distance is available, one could apply CP with the more conservative miscoverage level
\begin{equation} \tilde{\alpha}=\bigg[\alpha-\lVert p(x|c^\text{cal})-p(x|c^\text{te})\rVert_\text{TV}\bigg]^{+}, 
\end{equation}
where $[\cdot]^+$ denotes the positive part operator, i.e., $[\cdot]^+=\max(\cdot,0)$. By (\ref{tv}), this method can ensure the desired coverage level of $\alpha$, while generally increasing the inefficiency (\ref{inefficiency}).

\subsection{Weighted Conformal Prediction}
\label{weighted cp}
As discussed, the standard CP set construction (\ref{cp_shift}) suffers from the coverage gap in (\ref{cp_coverage}) in the presence of a covariate shift. WCP eliminates the coverage gap as long as one knows the covariate likelihood ratio 
\begin{equation}
w(x,c^\text{te},c^\text{cal})=\frac{p(x|c^\text{te})}{p(x|c^\text{cal})} 
\label{likelihood}
\end{equation}
for all calibration and test inputs $x$. Using the ratio (\ref{likelihood}), WCP accounts for the covariate shift between the test data and the calibration data set by weighting the scores according to their covariate likelihoods. Specifically, given the calibration data set $\mathcal{D}_{c^\text{cal}}$ for some context $c^\text{cal}$, WCP builds the prediction set as
    \begin{equation}
\Gamma_{\alpha}^\text{WCP}(x^\text{te}|c^\text{te},c^\text{cal})=\bigg\{y \in \mathcal{Y}: S(x^\text{te},y) \leq \textrm{Quantile}\bigg({1-\alpha};\sum_{i=1}^{|\mathcal{D}_{c^\text{cal}}|}\tilde{w}_{i}\delta_{S(x[i],y[i])}+\tilde{w}_{|\mathcal{D}_{c^\text{cal}}|+1}\delta_{\infty}\bigg)\bigg\},
\label{wcp}
    \end{equation}
    with self-normalizing weights 
    \begin{equation}
    \tilde{w}_{i}=\frac{w(x[i],c^\text{te},c^\text{cal})}{\sum_{i=1}^{|\mathcal{D}_{c^\text{cal}}|} w(x[i],c^\text{te},c^\text{cal})+w(x^\text{te},c^\text{te},c^\text{cal})}    
    \label{multi_weight}
    \end{equation}
     and 
     \begin{equation}
\tilde{w}_{|\mathcal{D}_{c^\text{cal}}|+1}=\frac{w(x^\text{te},c^\text{te},c^\text{cal})}{\sum_{i=1}^{|\mathcal{D}_{c^\text{cal}}|} w(x[i],c^\text{te},c^\text{cal})+w(x^\text{te},c^\text{te},c^\text{cal})}.
     \end{equation} The weights $\{\tilde{w}_1,\ldots \tilde{w}_{|\mathcal{D}^\text{cal}|+1}\}$ form a probability mass function that assigns a larger probability to data points $x[i]$ that are more likely under the test context $c^\text{te}$. 
    
WCP guarantees the target coverage (\ref{set_predictor}) i.e.,
\begin{equation}
        \text{Pr}\big[y^\text{te} \in \Gamma_{\alpha}^\text{WCP}(x^\text{te}|c^\text{te},c^\text{cal})\big]\geq 1-\alpha,
        \label{wcp_coverage}
    \end{equation}
where the expectation is taken with respect to the joint distribution (\ref{joint_dist_eq})-(\ref{c_cal_eq}) \cite{Tibshirani2019Dist_shift, Lei2021TV}. 

As discussed in \cite{bhattacharyya2024groupWCP}, the ratio (\ref{likelihood}) can be practically estimated if one has unlabeled data for the target and test contexts. In the considered setting, we assume that the controller has no data for the test context, making WCP inapplicable.

%However, in most situations, it is impossible to obtain the optimal covariate distribution in the unknown (unseen) test context, and the following sections introduce methods for estimating the covariate likelihood ratio using functions of the data and context variables.

\section{Meta-Learned Context-Dependent Weighted Conformal Prediction} \label{Meta-Learned Context-Dependent Weighted Conformal Prediction}
In this section, we introduce the proposed ML-WCP method. As depicted in Fig. \ref{fig2}, ML-WCP operates in two phases: a meta-training phase and a testing phase. During the \textit{offline} meta-training phase, ML-WCP optimizes a covariate likelihood ratio estimator $\omega_{\theta}(x,c_1,c_2)$ of the likelihood ratio (\ref{likelihood}). The estimator $\omega_{\theta}(x,c_1,c_2)$ makes it possible to approximate the covariate likelihood (\ref{likelihood}) for any pair of contexts without requiring any data for the two contexts $c_1$ and $c_2$. Using this model, for a given selected calibration context $c^\text{cal}$, ML-WCP applies the WCP set prediction (\ref{wcp}) with the estimate $\omega_{\theta}(x,c^\text{te},c^\text{cal})$ in lieu of the ground-truth ratio $w(x,c^\text{te},c^\text{cal})$. As we will prove, the coverage performance depends on the quality of the estimator $\omega_{\theta}(x,c^\text{te},c^\text{cal})$. %We formulate the problem of estimating parameterized weights using a meta-learning approach, propose a permutation-equivariant architecture, and present the overall algorithm. 
\subsection{Offline Meta-Learning}
\label{offline meta-learning}
As a first step, as seen in Fig. \ref{fig2}, ML-WCP partitions the set of contexts, $\mathcal{C}$, into a subset $\mathcal{C}^\text{tr}$ and a subset $\mathcal{C}^\text{cal}$ with $\mathcal{C}^\text{tr} \cup \mathcal{C}^\text{cal}=\mathcal{C}$ and $\mathcal{C}^\text{tr} \cap \mathcal{C}^\text{cal}=\emptyset$. As discussed in this section, the set $\mathcal{C}^\text{tr}$ of contexts is used to optimize the zero-shot estimator $\omega_{\theta}(x,c_1,c_2)$, while the set $\mathcal{C}^\text{cal}$ is leveraged for calibration. 

During the offline meta-learning phase, ML-WCP optimizes a zero-shot estimator of the weights required by the WCP set (\ref{wcp_coverage}). Specifically, as shown in Fig. \ref{fig2}, ML-WCP leverages training data logged from multiple contexts $c\in \mathcal{C}^\text{tr}$ to estimate a mapping $w_{\theta}(x,c_1,c_2)$ that approximates the covariate likelihood ratio weights (\ref{likelihood}) for any input $x$ and context pair $(c_1,c_2)$.

To proceed, a first observation is that the WCP prediction set (\ref{wcp_coverage}) depends on the covariate likelihood ratio (\ref{likelihood}) only through the normalized weights $\{\tilde{w}_i\}_{i=1}^{|\mathcal{D}_{c^\text{cal}}|}$ in (\ref{wcp}), which are invariant to any positive scaling of the ratio $w(w,c_1,c_2)$. The following lemma identifies a quantity proportional to the covariate likelihood ratio $w(x,c_1,c_2)$ that can be conveniently estimated.
\begin{lemma}
    For any two contexts $c_1 \in \mathcal{C}$ and $c_2 \in \mathcal{C}$ with $c_1 \neq c_2$, the weight function $w(x,c_1,c_2)$ can be expressed as
    \begin{equation}
        w(x,c_1,c_2)=\frac{q(c_1|x)}{q(c_2|x)},
        \label{w_express}
    \end{equation}
    where 
    \begin{equation}
    q(c_i|x)=\frac{p(x|c_i)}{p(x|c_1)+p(x|c_2)}, \quad \textnormal{for } i=1,2.
    \label{q_express}
    \end{equation}
\end{lemma}
This result follows directly by substituting (\ref{q_express}) into (\ref{w_express}). Based on Lemma 1, ML-WCP estimates the ratio (\ref{w_express}). This can be done by observing that the term $q(c_i|x)$ in (\ref{q_express}) is the posterior distribution obtained from the joint distribution $q(c_i,x)=p(c_i)p(x|c_i)$ with $p(c_i)=1/2$ for $i=1,2$. Thus, one can estimate $q(c_i|x)$ by training a classifier to distinguish contexts $c_1$ and $c_2$ based on the observation $x$.

Accordingly, ML-WCP adopts a parametric estimator $\omega_{\theta}(x,c_1,c_2)$, whose parameters $\theta$ are optimized during the meta-training phase to distinguish between pairs of contexts $c_1$ and $c_2$. For a randomly selected pair of context vectors $c_1\neq c_2$ and $c_1, c_2 \in \mathcal{C}^\text{tr}$, ML-WCP constructs a labeled data set $\mathcal{D}_{c_1,c_2}$ by drawing $D$ examples from data set $\mathcal{D}_{c_1}$ and $D$ examples from data set $\mathcal{D}_{c_2}$. Each data point $(x[i],y[i])$ from data set $\mathcal{D}_{c_1}$ is labeled as $z[i]=1$, while all examples $(x[i],y[i])$ in data set $\mathcal{D}_{c_2}$ are labeled as $z[i]=0$. This way, detecting label $z$ also identifies the context vector from the given pair $(c_1,c_2)$. Considering all the $2D$ data points results in the data set $\mathcal{D}_{c_1,c_2}=\{(x[i],y[i],z[i])\}_{i=1}^{2D}$.

To train a binary classifier $\omega_{\theta}(x,c_1,c_2)$ to distinguish between contexts $c_1$ and $c_2$, we adopt the standard cross-entropy loss
\begin{equation}
\begin{aligned}
        \mathcal{L}_{c_1,c_2}(\theta)&=-\sum_{i=1}^{2D}\bigg[z[i]\log\big(q(c_1|x)\big)+(1-z[i])\log \big(q(c_2|x)\big)\bigg]\\
        &=-\sum_{i=1}^{2D}\bigg[z[i]\log\bigg(\frac{w_{\theta}(x,c_1,c_2)}{1+w_{\theta}(x,c_1,c_2)}\bigg)+(1-z[i])\log \bigg(\frac{1}{1+w_{\theta}(x,c_1,c_2)}\bigg)\bigg]\\&=\sum_{i=1}^{2D}\bigg[-z[i]\log\big(\omega_{\theta}(x,c_1,c_2)\big)+\log\big(1+\omega_{\theta}(x,c_1,c_2)\big)\bigg].
        \label{loss}            
\end{aligned}
    \end{equation}
For meta-training, we randomly sample $M$ pairs $(c_1,c_2)\in \mathcal{C}^\text{tr}\times \mathcal{C}^\text{tr}$, and combine the resulting losses (\ref{loss}) by summing the contributions of all sampled pairs $(c_1,c_2)$. The optimization of the sum-loss is carried out via stochastic gradient descent \cite{simeone2022machine}. The meta-learning procedure adopted by ML-WCP is outlined in Algorithm 1. 
\begin{algorithm}
% \textbf{Input:} Meta-training data set $\mathcal{D}^\text{m-tr}$; step size hyperparameter $\kappa$\\
% \textbf{Output:} (Meta-learned) parameter vector $\theta^*$\\ \vspace{-13pt}
% \hrulefill\\
\textbf{Initialize} parameter vector $\theta$, meta-training data set $\{\mathcal{D}_{c}\}_{c \in \mathcal{C}^\text{tr}}$, step size hyperparameter $\kappa$, maximum number of context pairs $M$
\vspace{-13pt}\\
\hrulefill \\ \vspace{-3pt}
\textbf{\textit{Meta-Learning Phase}}\\ \vspace{-13pt}
\hrulefill \vspace{-3pt}\\ 
\While{\normalfont{convergence criterion not met}}{
    \For{$m=1,\ldots,M$}{Randomly select a context $c_1$ and context $c_2 \neq c_1$, and draw $D$ examples from data sets $\mathcal{D}_{c_1}$ and $\mathcal{D}_{c_2}$\\
    Each data point $(x[i],y[i])$ from data set $\mathcal{D}_{c_1}$ is labeled as $z[i]=1$, while those in $\mathcal{D}_{c_2}$ are labeled as $z[i]=0$\\
    With labeled data set $\mathcal{D}_{c_1,c_2}=\{(x[i],y[i],z[i])\}_{i=1}^{2D}$, compute the cross-entropy loss (\ref{loss})
    }
    Update the parameter vector as $\theta\leftarrow \theta-\kappa\sum_{c_1,c_2}\nabla\mathcal{L}_{c_1, c_2}(\theta)$, where the sum is extended to a mini-batch of contexts.
        
}
Return the optimized parameter vector $\theta$ \vspace{-10pt}\\ \hrulefill \\
\vspace{-5pt}
\textbf{\textit{Runtime Phase}}\\ \vspace{-13pt}
\hrulefill \vspace{-3pt} \\
Given a test context $c^\text{te}$, select a calibration context $c^\text{cal}\in \mathcal{C}^\text{cal}$ by (\ref{single_cal_select}), and return prediction set $\Gamma_{\alpha}^\textnormal{ML-WCP}(x^\textnormal{te}|c^\textnormal{te},c^\textnormal{cal})$.
\label{alg}
\caption{ML-WCP}
\end{algorithm}

\subsection{Model Architecture} 
\label{model architecture}
As explained above, ML-WCP trains a parametric classifier $\omega_{\theta}(x,c_1,c_2)$ that can automatically adapts to different context pairs $(c_1,c_2)$. By definition, the covariate likelihood ratio estimator $\omega_{\theta}(x,c_1,c_2)$ must satisfy the structural property
\begin{equation}
    \omega_{\theta}(x,c_2,c_1)=\frac{1}{\omega_{\theta}(x,c_1,c_2)}.
    \label{architecture_eq}
\end{equation}
In fact, the covariate likelihood ratio satisfies the corresponding property
\begin{equation}
    \omega(x,c_2,c_1)=\frac{q(x|c_2)}{q(x|c_1)}=\frac{1}{\omega(x,c_1,c_2)}.
\end{equation}
To ensure this property, we propose to adopt the architecture in Fig. \ref{fig2}, which expresses the estimator as
\begin{equation}
\omega_{\theta}(x,c_1,c_2)=\exp \big(g_{\theta}(x,c_1)-g_{\theta}(x,c_2)\big),   
\label{architecture}
\end{equation}
where $g_{\theta}(x,c)$ is a parametric function such as a multi-layer perceptron (MLP). One can readily check that the parameterization (\ref{architecture}) satisfies the property (\ref{architecture_eq}) for any function $g_{\theta}(x,c)$.

\subsection{Calibration Guarantees}
The proposed ML-WCP algorithm summarized in Algorithm 1 satisfies the following coverage guarantee.
\begin{lemma}
ML-WCP satisfies the coverage guarantee
\begin{equation}
\textnormal{Pr}\big[y^\textnormal{te} \in \Gamma_{\alpha}^\textnormal{ML-WCP}(x^\textnormal{te}|c^\textnormal{te},c^\textnormal{cal})\big]\geq 1-\alpha-\frac{1}{2}\mathbb{E}_{x \sim p(x|c^\textnormal{cal})}\bigg|\frac{\omega_{\theta}(x,c^\textnormal{te},c^\textnormal{cal})}{\mathbb{E}_{x \sim p(x|c^\textnormal{cal})}[\omega_{\theta}(x,c^\textnormal{te},c^\textnormal{cal})]}-w(x,c^\textnormal{te},c^\textnormal{cal})\bigg|,
\label{mlwcp_coverage}
    \end{equation}
where the probability on the left-hand side is computed over joint distribution (\ref{joint_dist_eq}), while the average on the right-hand side is taken with respect to the distribution $p(x|c^\textnormal{cal})$ of input $x$ given the calibration context $c^\textnormal{cal}$.
\end{lemma}
\begin{proof}
The proof follows from  \cite[Theorem 1]{Lei2021TV} and detailed in Appendix A.
\end{proof}

\section{ML-WCP with Multi-Context Calibration}
\label{ML-WCP with Multi-Context Calibration}
So far, we have assumed the use of data from a single context $c^\text{cal}$ for calibration. In practice, data from any given context may be limited, motivating the use of data from multiple contexts for calibration. This section explores this extension. Specifically, after proposing two general approaches for selecting the subset of calibration contexts, we introduce two multi-context ML-WCP schemes.%moving beyond the ML-WCP in Section \ref{Meta-Learned Context-Dependent Weighted Conformal Prediction} reliance on selecting a single calibration context for a given test context $c^\text{te}$. This extended framework presents the criteria for how to select a calibration context set based on the distance-based selection mentioned in Section \ref{runtime calibration}. Then adapt the zero-shot estimator to learn the covariate likelihood ratio between a single context $c_1$ and a set of contexts $\mathcal{C}_2$ from $\mathcal{C}^\text{tr}$ in the same way. Compared to conventional ML-WCP, inefficiency performance depends on the quality of the criterion \textcolor{red}{(how to choose $\epsilon$)} used to select the calibration context set, while still satisfying the coverage guarantee.

\subsection{Selecting Multi-Context Calibration Data}

Given the test context $c^\text{te}$, we start by selecting a subset of calibration contexts from the overall calibration set $\mathcal{C}^\text{cal}$. By leveraging a distance-based metric as in Section \ref{calibration data}, a calibration subset $\tilde{\mathcal{C}}^\text{cal} \subseteq \mathcal{C}^\text{cal}$ can be obtained in one of the following ways.

\subsubsection{Fixed number of calibration contexts}
This approach chooses a fixed number $K^\text{cal}$ of calibration contexts by finding the subset $\tilde{\mathcal{C}}^\text{cal} \subseteq \mathcal{C}^\text{cal}$ of vectors $c^\text{cal}$ that closet to $c^\text{te}$:
\begin{equation}
    \tilde{\mathcal{C}}^\text{cal}=\argminB_{\tilde{\mathcal{C}}^\text{cal}} \sum_{c \in \tilde{\mathcal{C}}^\text{cal}} d(c^\text{te},c), \,\,\, \text{s.t. } |\tilde{\mathcal{C}}^\text{cal}|=K^\text{cal}.
    \label{multi_cal_select_fixed}
\end{equation}
% \subsubsection{Fixed Number of Calibration Contexts}
% To ensure that the same number $K^\text{cal}$ of calibration contexts are chosen for each test context $c^\text{te}$, the set $\tilde{\mathcal{C}}^\text{cal}$ includes the $K^\text{cal}$ contexts $c \in \mathcal{C}^\text{cal}$ with the $K^\text{cal}$ smallest distances $d(c^\text{te},c^\text{cal})$ as
% \begin{equation}
%     \tilde{\mathcal{C}}^\text{cal}=\{c^{(1)},c^{(2)},\ldots,c^{(K^\text{cal})}\},\,\, \textnormal{where }d(c^\text{te},c^{(i)})<d(c^\text{te},c^{(j)}) \textnormal{ for all }i<j. 
% \end{equation}
% Note that this recovers the selection in Section \ref{distance-based selection} for $K^\text{cal}=1$. 
\subsubsection{Adaptive number of calibration contexts}
Different test contexts $c^\text{te}$ may require different amount of calibration data. Furthermore, different contexts may have varying numbers of relevant calibration contexts. Therefore, it may be preferable to choose the subset $\tilde{\mathcal{C}}^\text{cal}\subseteq \mathcal{C}^\text{cal}$ so as to include all the contexts $c \in \tilde{\mathcal{C}}^\text{cal}$ that are sufficiently close to the test context $c^\text{te}$. This yields the subset 
\begin{equation}
    \tilde{\mathcal{C}}^\text{cal}=\{c \in \mathcal{C}^\text{cal}:d(c^\text{te},c)\leq \epsilon \},
    \label{multi_cal_select_adapt}
\end{equation}
for some threshold $\epsilon>0$. With this approach, the number of selected calibration contexts $|\tilde{\mathcal
{C}}^\text{cal}|=K^\text{cal}$ is adapted to the given test context $c^\text{te}$. Note that, in practice, one can ensure that the set $\tilde{\mathcal{C}}^\text{cal}$ has at least one context through a suitable choice of the threshold $\epsilon$.% the threshold $\epsilon$ ensures the calibration subset $\tilde{\mathcal{C}}^\text{cal}$ has $K^\text{cal}$ contexts. In practice, one can true the value of $\epsilon$ such that the set $\tilde{\mathcal{C}}^\text{cal}$ has at least one context.
\subsection{ML-WCP with Majority Vote}
A simple way to leverage multiple calibration contexts is to construct the ML-WCP prediction set separately for each calibration context $c \in \tilde{\mathcal{C}}^\text{cal}$, as discussed in the previous section, and then combine the resulting sets $\{\Gamma_{\alpha}^\text{ML-WCP}(x^\text{te}|c^\text{te},c)\}_{c \in \tilde{\mathcal{C}}^\text{cal}}$ in (\ref{mlwcp_coverage}). This can be done via a majority vote-based approach by including all output values $y\in\mathcal{Y}$ included by a least half of the sets $\{\Gamma_{\alpha}^\text{ML-WCP}(x^\text{te}|c^\text{te},c)\}_{c \in \tilde{\mathcal{C}}^\text{cal}}$ \cite{ramdas2024majorityvote}. This yields the set
\begin{equation}
\Gamma_{\alpha}^\text{ML-WCP-MV}(x^\text{te}|c^\text{te},\tilde{\mathcal{C}}^\text{cal}) = \bigg\{y \in \mathcal{Y}: \frac{1}{K^\text{cal}}\sum_{c \in \tilde{\mathcal{C}}^\text{cal}}\mathbbm{1}\big(y \in \Gamma_{\alpha}^\text{ML-WCP}(x^\text{te}|c^\text{te},c)\big) >\frac{1+u}{2}\bigg\},    
\end{equation}
where $u$ is a realization of a uniform random variable on the interval $[0,1]$. 

ML-WCP-MV satisfies the following coverage guarantee.

\begin{lemma}
    ML-WCP-MV satisfies the following coverage property
\begin{equation}
\label{mlwcp_mv_coverage}
\textnormal{Pr}\big[y^\textnormal{te} \in \Gamma_{\alpha}^\textnormal{ML-WCP-MV}(x^\textnormal{te}|c^\textnormal{te},\tilde{\mathcal{C}}^\textnormal{cal})\big] \geq 1-2 \frac{\sum_{k=1}^{K^\textnormal{cal}} \alpha_k}{K^\textnormal{cal}},
\end{equation}
where
\begin{equation}
     \alpha_k=\alpha+\frac{1}{2}\mathbb{E}_{x \sim p(x|c_k)}\bigg|\frac{\omega_{\theta}(x,c^\textnormal{te},c_k)}{\mathbb{E}_{x \sim p(x|c_k)}[\omega_{\theta}(x,c^\textnormal{te},c_k)]}-w(x,c^\textnormal{te},c_k)\bigg|
\end{equation}
is the miscoverage level \textnormal{(\ref{mlwcp_coverage})} for the $k$-th calibration context $c_k$ from obtained $\tilde{\mathcal{C}}^\textnormal{cal}$.
\end{lemma}
\begin{proof}
    The proof follows from \cite[Proposition 11.15]{ramdas2024majorityvote}.
\end{proof}
Therefore, combining multiple subsets can enhance the inefficiency (\ref{inefficiency}), although this can come at the expense of a reduction in the coverage. For instance, if all the combined sets have the same coverage level $\alpha_k=\alpha$, the bound (\ref{mlwcp_mv_coverage}) yields a doubling of the miscoverage rate.

%We note that, to ensure consistent comparison with other methods targeting the miscoverage level $\alpha$, we employ a prediction set adjusted to a stricter miscoverage level of $\alpha/2$ to closely guarantee coverage level $1-\alpha$.

\subsection{Multi-Context ML-WCP via Mixing}
\begin{figure}[h]
    \centering
    \includegraphics[width=0.8\linewidth]{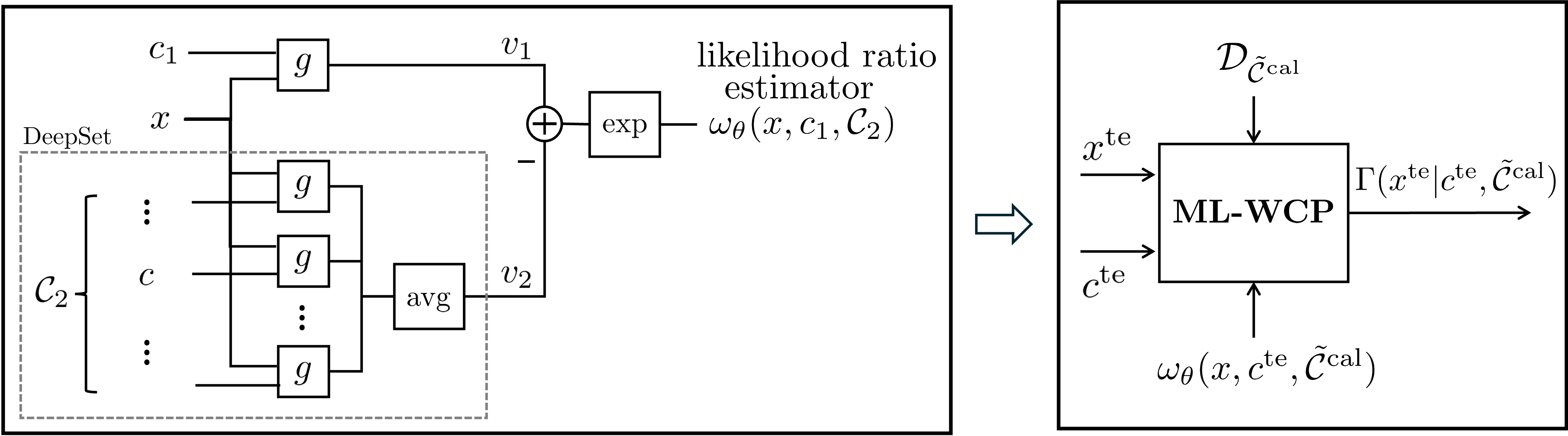}
    \caption{Multi-Context ML-WCP via Mixing: likelihood ratio estimator (left) and set predictor (right).}
    \label{mlwcp_mix}
\end{figure}

ML-WCP-MV uses data from different calibration contexts separately, combining the resulting prediction. A potentially more efficient approach is presented in this section that treats the calibration data set $\mathcal{D}_{\tilde{\mathcal{C}}^\text{cal}}=\cup_{c \in \tilde{\mathcal{C}}^\text{cal}}\mathcal{D}_c$ from the selected contexts $\tilde{\mathcal{C}}^\text{cal}$ jointly. The main idea behind the approach is to model the data set $\mathcal{D}_{\tilde{\mathcal{C}}^\text{cal}}$ as being generated from a mixture distribution.

Specficially, the distribution of the calibration data is assumed to be given by
\begin{equation}
\label{mixture_dist}
    \frac{1}{|\tilde{\mathcal{C}}^\text{cal}|}\sum_{c \in \tilde{\mathcal{C}}^\text{cal}}p(x|c).
\end{equation}
Accordingly, one first selects a context from the set $\tilde{\mathcal{C}}^\text{cal}$ with equal probability, and then samples a data point corresponding to the chosen context. Using the mixture distribution (\ref{mixture_dist}), the resulting covariate likelihood ratio $w(x,c^\text{te},\tilde{\mathcal{C}}^\text{cal})$ between the test context $c^\text{te}$ and the set of calibration contexts $\tilde{\mathcal{C}}^\text{cal}$ is
\begin{equation}
    w(x,c^\text{te},\tilde{\mathcal{C}}^\text{cal}) = \frac{p(x|c^\text{te})}{\frac{1}{|\tilde{\mathcal{C}}^\text{cal}|}\sum_{c\in\tilde{\mathcal{C}}^\text{cal}}p(x|c)},
    \label{multi-context likelihood}
\end{equation}
where the denominator reflects the mixture distribution (\ref{mixture_dist}) over the calibration contexts $c \in \tilde{\mathcal{C}}^\text{cal}$. %ML-WCP-Mix produces the set (\ref{wcp}), which is guaranteed to satisfy the condition (\ref{mlwcp_coverage}) with respect to the joint distribution of test and calibration data given the respective context vectors.
 
We propose to estimate the covariate likelihood ratio $w_{\theta}(x,c_1,\mathcal{C}_2)$ for any input $x$ and pair $(c_1,\mathcal{C}_2)$ of context $c_1$ and subset of contexts $\mathcal{C}_2$ using meta-learning. In a manner similar to Fig. \ref{fig1}, the following lemma identifies a quantity proportional to the covariate likelihood ratio \( w(x, c_1, \mathcal{C}_2) \) that can be conveniently estimated.

\begin{lemma}
For any single context \( c_1 \in \mathcal{C} \) and subset of contexts \( \mathcal{C}_2 \subseteq \mathcal{C} \), where \( c_1 \notin \mathcal{C}_2 \), the weight function \( w(x, c_1, \mathcal{C}_2) \) in (\ref{multi-context likelihood}) can be expressed as
\begin{equation}
w(x, c_1, \mathcal{C}_2)=\frac{{q(c_1|x)}}{\frac{1}{|\mathcal{C}_2|} \sum_{c \in \mathcal{C}_2} {q(c|x)}},
\label{w_express_multicontext}
\end{equation}
where 
\begin{equation}
    q(c|x)=\frac{p(x|c)}{p(x|c_1)+\sum_{c \in \mathcal{C}_2}p(x|c)}.
    \label{q_express_multicontext}
\end{equation}
\end{lemma}
This result follows directly by substituting (\ref{q_express_multicontext}) into (\ref{w_express_multicontext}). Based on Lemma 2, we propose to estimate the ratio (\ref{w_express_multicontext}) via meta-learning. This can be done by observing that the term $q(c|x)$ in (\ref{q_express_multicontext}) is the posterior distribution obtained from the joint distribution $q(c,x)=p(c)p(x|c)$, with $p(c)=1/(1+|\mathcal{C}_2|)$ for $c \in \{c_1\}\cup \mathcal{C}_2$. Thus, we can estimate $q(c|x)$ by training a classifier to distinguish data generated under context $c_1$ against data generated under the mixture distribution (\ref{mixture_dist}).
% \textcolor{red}{In fact, as in Lemma 1, the process of calculating $w(x,c_1,\mathcal{C}_2)$ in this case proceeds similarly and can be expressed as follows: 
% \begin{equation}
%     w(x,c_1,\mathcal{C}_2)=\frac{q(c_1|x)/q(c_1)}{\dfrac{1}{|\mathcal{C}_2|}\sum_{c\in\mathcal{C}_2}q(c|x)/q(c)},
% \end{equation}
% In this case, it seems that we can assume that factors such as the terms of $q(c)$ are independent of $x$ and therefore do not affect the weights $\{\tilde{w}_i\}_{i=1}^{|\mathcal{D}_{\tilde{\mathcal{C}}^\text{cal}}|}$. Is it right?}

The proposed method adopts a parametric estimator $\omega_{\theta}(x,c_1,\mathcal{C}_2)$, whose parameter $\theta$ is optimized during the meta-training phase to distinguish between distribution $p(x|c_1)$ and (\ref{mixture_dist}). For a randomly selected context subset $\mathcal{C}_2$ and $c_1 \notin \mathcal{C}_2$ and $\{c_1\}, \mathcal{C}_2 \subseteq \mathcal{C}^\text{tr}$, we construct a labeled data set $\mathcal{D}_{c_1,\mathcal{C}_2}$ by drawing $D$ examples from the data set $\mathcal{D}_{c_1}$ and $D/|{\mathcal{C}_2}|$ examples from each data set $\mathcal{D}_{c}$, where $c \in \mathcal{C}_2$, thereby ensuring a balanced data set $\mathcal{D}_{c_1}\cup \mathcal{D}_{\mathcal{C}_2}$ with $2D$ data points. Then, as in Section \ref{offline meta-learning}, each data point $(x[i],y[i])$ in data set $\mathcal{D}_{c_1}$ is labeled $z[i]=1$, while other data points in $\mathcal{D}_{\mathcal{C}_2}$ are labeled as $z[i]=0$. Detecting label $z$ identifies whether the result corresponds to $c_1$ or $\mathcal{C}_2$, and results in constructing the labeled data set $\mathcal{D}_{c_1,\mathcal{C}_2}=\{(x[i],y[i],z[i])\}_{i=1}^{2D}$.

To train a binary classifier, the cross-entropy loss is evaluated as
\begin{equation}
        \mathcal{L}_{c_1,\mathcal{C}_2}(\theta)=\sum_{i=1}^{2D}\bigg[-z[i]\log\big(\omega_{\theta}(x,c_1,\mathcal{C}_2)\big)+\log\big(1+\omega_{\theta}(x,c_1,\mathcal{C}_2)\big)\bigg].
        \label{loss multi context}
    \end{equation}
For meta-training, similar to Section \ref{offline meta-learning}, we randomly sample $M$ pairs $(c_1,\mathcal{C}_2) \in \mathcal{C}^\text{tr} \times \mathcal{C}^\text{tr}$. The resulting sum-loss, in averaging (\ref{loss multi context}) over all the sampled pairs $(c_1, \mathcal{C}_2)$, is optimized via stochastic gradient descent.
%\textcolor{red}{먼저 random으로 context $c_1$와 random한 context 수 $k_2$를 가진 set $\mathcal{C}_2$를 sample한다. 선택된 set $\mathcal{C}_2$에 해당하는 data set $\mathcal{D}_{\mathcal{C}_2}$ 또한 mixture distribution으로부터 형성된다. Specifically, one first selects a context from the set $\mathcal{C}_2$ with equal probability, and then samples a data point corresponding to the chosen context. 이 과정을 반복하여 data set $\mathcal{D}_{\mathcal{C}_2}$를 형성한다.} Afterwards, each data point $(x[i],y[i])$ in data set $\mathcal{D}_{c_1}$ is labeled as $z[i]=1$, while those in $\mathcal{D}_{\mathcal{C}_2}$ are labeled as $z[i]=0$. 
% Overall data set generation and the procedure of optimizing the parameter $\theta$ through zero-shot meta-learning is outlined in Algorithm 1. 

\subsubsection{Model Architecture}
The covariate likelihood ratio $\omega(x,c_1,\mathcal{C}_2)$ has the property of being permutation-invariant with respect to the elements of the set $\mathcal{C}_{2}$. 
To ensure this condition, we adopt the architecture as shown in Fig. \ref{mlwcp_mix}. Accordingly, following a DeepSet architecture \cite{zaheer2018deepsets}, all input vectors share the parametric function $g_{\theta}(x,c)$ introduced in (\ref{architecture}), and we set
\begin{equation}
    \omega_{\theta}(x,c_1,\mathcal{C}_2)=\exp\bigg(g_\theta(x,c_1)-\frac{1}{|\mathcal{C}_2|}\sum_{c \in \mathcal{C}_2}g_{\theta}(x,c)\bigg).
    \label{multi_parametric_ft}
\end{equation}

\section{Simulation results} 
\label{Simulation results}
In this section, we present experimental results to validate the performance of ML-WCP. We consider the tasks described in Section \ref{ex_wireless communication}, encompassing the application-layer problem of traffic classification, the task of profiling a third-party scheduling app at the MAC layer, and the phsical-layer problem of list decoding.
\subsection{Benchmarks}
We consider as benchmarks the following schemes:
\begin{itemize}
    \item \textit{Top-$K$ prediction}, which generates a prediction set encompassing the labels that correspond to the $K$ highest scores of the pre-trained probabilistic model \cite{Angelopoulos2020Top-k};
    \item \textit{CP}, described in Section \ref{Conformal Prediction}, which disregards the covariate shift between calibration and testing data sets;
    \item \textit{CCP}, discussed in Section \ref{conservative conformal prediction}, which applies CP with a more conservative miscoverage target based on an estimate of the covariate likelihood ratio;
    \item \textit{Ideal WCP}, which uses the the true covariate likelihood ratio $w(x,c^\text{te},c^\text{cal})$;
    \item \textit{ML-WCP}, which we have introduced in Section \ref{Meta-Learned Context-Dependent Weighted Conformal Prediction}.
\end{itemize}

Note that Ideal WCP is only applicable in an ideal scenario with full distributional information so that the true covariate likelihood ratio can be evaluated. Furthermore, to implement CCP, we note that the TV distance in (\ref{tv}) can be expressed as
\begin{equation}
\begin{aligned}
    \lVert p(x|c^\text{te})-p(x|c^\text{cal})\rVert_\text{TV}&=\frac{1}{2}\int| p(x|c^\text{te})-p(x|c^\text{cal})\big|\mathrm{d}x\\&=\frac{1}{2}\int\bigg|\frac{p(x|c^\text{te})}{p(x|c^\text{cal})} -1\bigg|p(x|c^\text{cal})\mathrm{d}x\\&=\frac{1}{2}\mathbb{E}_{x \sim p(x|c^\text{cal})}\big| w(x,c^\text{te},c^\text{cal})-1 \big|.
\end{aligned}
\end{equation}
Thus, the TV distance can be estimated by replacing the covariate likelihood ratio $w(x,c^\text{te},c^\text{cal})$ with the model $\omega_{\theta}(x,c^\text{te},c^\text{cal})$, which is meta-trained as discussed in Section \ref{offline meta-learning}. 
% , we can calculate the tv distance with our estimated $\omega_{\theta}$, i.e.,

\subsection{Implementation}
For the single-context ML-WCP, introduced in Section \ref{Meta-Learned Context-Dependent Weighted Conformal Prediction}, as well as for CCP, the covariate likelihood ratio estimator $\omega_{\theta}(x,c_1,c_2)$ is implemented via the model (\ref{architecture_eq}), in which function $g_{\theta}(x,c)$ is an MLP consisting of a fully connected network, followed by four hidden layers with ReLU activation. To minimize the loss (\ref{loss}), we use an Adam optimizer \cite{Kingma2014adam} with learning rate $\eta=0.001$ and weight decay $\lambda=10^{-5}$.

For the multi-context ML-WCP scheme presented in Section \ref{ML-WCP with Multi-Context Calibration}, ML-WCP-MV follows the same settings as the single-context ML-WCP, while ML-WCP-Mix adopts the function $g_{\theta}(x,c)$ in (\ref{multi_parametric_ft}) with the same configuration as the single-context ML-WCP and minimizes the loss (\ref{loss multi context}) with learning rate $\eta=0.005$. 

For all applications, the score $S(x[i],y[i])$ assigned to an input-output pair $(x[i],y[i])$ is the negative log-likelihood probability of the corresponding predictive model.
\subsection{Evaluation}
We adopt as performance measures the empirical coverage
and empirical inefficiency, which are defined as the empirical estimates of the metrics (\ref{set_predictor}) and (\ref{inefficiency}), respectively, as
\begin{equation}
   \text{Empirical coverage}=\frac{1}{|\mathcal{D}_{c^\text{te}}|}\sum_{i=1}^{|\mathcal{D}_{c^\text{te}}|}\mathbbm{1}(y^\text{te}[i]\in \Gamma_{\alpha}\big(x^\text{te}[i]|c^\text{te},\mathcal{C}^\text{cal})\big)
    \label{empir_cov}
\end{equation}
and
\begin{equation}
    \text{Empirical inefficiency}=\frac{1}{|\mathcal{D}_{c^\text{te}}|}\sum_{i=1}^{|\mathcal{D}_{c^\text{te}}|}\big|\Gamma_{\alpha}(x^\text{te}[i]|c^\text{te},\mathcal{C}^\text{cal})\big|,
    \label{empir_ineff}
\end{equation} 
based on a test data set $\mathcal{D}_{c^\text{te}}=\{c^\text{te},x[i],y[i]\}_{i=1}^{|\mathcal{D}_{c^\text{te}}|}$. 

For the case of a single calibration context $\mathcal{C}^\text{cal}=\{c^\text{cal}\}$, unless stated otherwise, we average the empirical coverage (\ref{empir_cov}) and empirical inefficiency (\ref{empir_ineff}) over independent draws of all pairs of context-dependent calibration and test data set $\{\mathcal{D}_{c^\text{te}},\mathcal{D}_{c^\text{cal}}\}$, where $(c^\text{te},c^\text{cal})\in\mathcal{C}^\text{tr}\times \mathcal{C}^\text{tr}$, and $c^\text{te}\neq c^\text{cal}$. We also consider the optimized selection (\ref{single_cal_select}) of single calibration contexts and (\ref{multi_cal_select_fixed}), (\ref{multi_cal_select_adapt}) of multiple calibration contexts, where $d(c_1,c_2)=1-(c_1 \cdot c_2)/(\lVert c_1 \rVert \lVert c_2 \rVert )$ \cite{huang2008cossimilar}. When not specified otherwise, ML-WCP refers to the single calibration context version described in Section \ref{Meta-Learned Context-Dependent Weighted Conformal Prediction}.

The simulation results are visualized using the standard box plot method \cite{krzywinski2014boxplot}. The boxes represent the 25\% (lower edge), 50\% (solid line within the box), and 75\% (upper edge) percentiles of the empirical performance metrics evaluated over different experiments, with the average value shown with a star marker.

% \begin{equation}
%         w(x,c_1,c_2)=\frac{q(c_1|x)}{q(c_2|x)}\cdot \frac{q(c_2)}{q(c_1)} \propto \frac{q(c_1|x)}{q(c_2|x)},
%         \label{w_express}
%     \end{equation}
%     \begin{equation}
%     \omega(x,c_1,c_2)=\frac{q(c_1|x)}{q(c_2|x)}.
% \end{equation}

% \subsection{Synthetic-Data Experiment}

\subsection{Traffic Slice Prediction} 
For the application-level traffic prediction introduced in Section \ref{ex_traffic prediction}, we leverage a real-world 5G dataset described in \cite{Chowdhury2024Megatron,Chowdbury[2024]TrafficPred.}. The goal is to calibrate a transformer-based traffic slice classifier xApp, assigning the observed traffic trace to one of $|\mathcal{Y}|=4$ labels: eMBB, mMTC, URLLC, or control (ctrl) traffic \cite{Petar2018Traffic}. The input $x$ consists of sequential KPIs processed in groups of $16$ consecutive time samples per KPI, which are treated as tokens fed into the transformer encoder layer \cite{Ashist2017Attention}. The resulting representations are flattened and fed to a fully connected layer with $256$ neurons
followed by ReLU activation. 
The transformer model hyperparameters follow the settings in \cite{Chowdhury2024Megatron} and the final accuracy of the prediction model we chose reaches $82.2\%$.

For each traffic slice, the data set provided by \cite{Chowdbury[2024]TrafficPred.} includes real-world contextual information such as location and mobility, which are designated as the context vector $c$ for use in our context-dependent set predictor framework. Accordingly, the context set $\mathcal{C}$ is given by the Cartesian product $\mathcal{C}=\mathcal{C}_0\times \mathcal{C}_1$, with $\mathcal{C}_0=\{\text{residential, campus, mixed}\}$, and $\mathcal{C}_1=\{\text{stationary, driving, walking}\}$. Each experiment samples $|\mathcal{D}_c|=1000$ data points for each context not used for training to obtain calibration and test pairs.

\begin{figure*}[t]
   \centering
   \hspace*{+0.75cm}
    \includegraphics[width=1.05\linewidth]{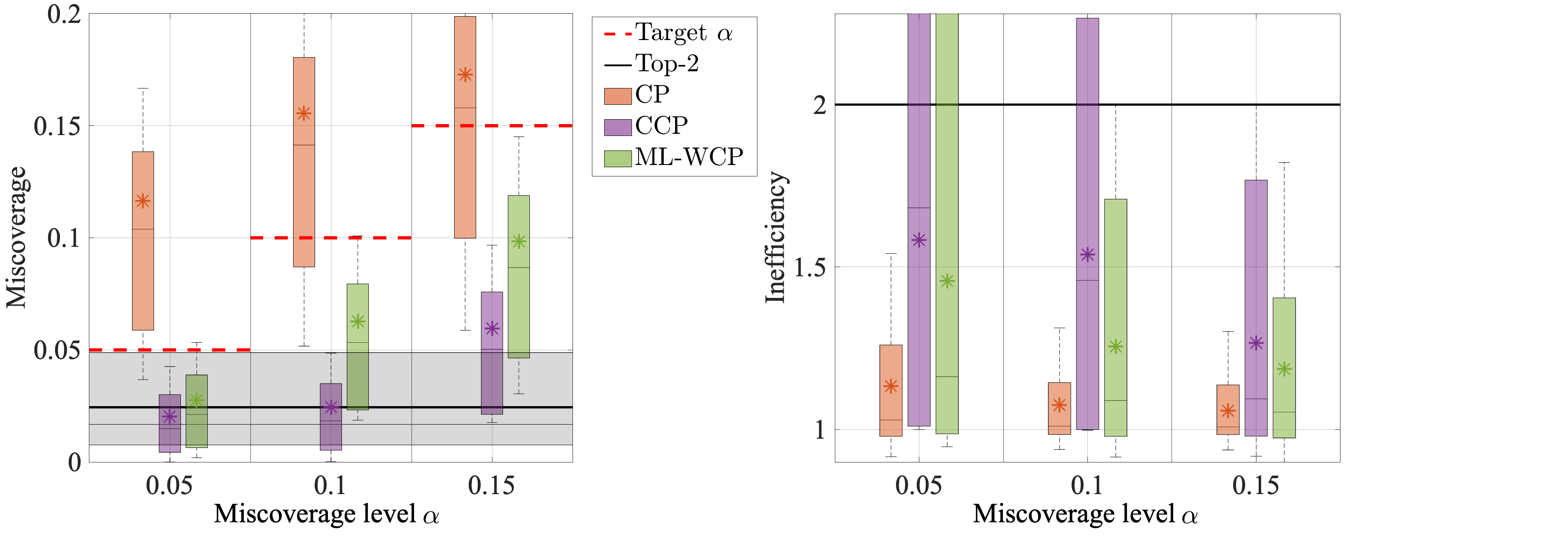}
    \caption{Empirical coverage and inefficiency of Top-2 prediction, CP, CCP, and ML-WCP versus the miscoverage level $\alpha$ with the number of meta-learning contexts $|\mathcal{C}^\text{tr}|=4$.}
    \label{app1_alpha}
\end{figure*}

Fig. \ref{app1_alpha} examines the performance of Top-2 prediction, CP, CCP, and ML-WCP as a function of the miscoverage level $\alpha$ with the number of meta-learning contexts fixed at $|\mathcal{C}^\text{tr}|= 4$. The analysis is conducted with calibration context $c^\text{cal}=(\textnormal{mixed, driving})$ and test context $c^\text{te}=(\textnormal{residential, stationary})$. Note that for the Top-2 method, the average miscoverage level, and its associated range of values represented as a gray bar, are constant as a function of $\alpha$. The results indicate that CP fails to achieve the target coverage level under covariate shift, whereas CCP and ML-WCP guarantee coverage within the desired level $\alpha$.

With lower miscoverage levels, CCP and ML-WCP exhibit high inefficiency variabilities. This observation aligns with the findings in \cite{Tibshirani2019Dist_shift}, which suggest that prediction sets based on covariate likelihood ratios prioritize coverage guarantees, often requiring larger set sizes in certain configurations to achieve reliability. However, as the miscoverage requirement is relaxed, i.e., $\alpha\geq0.1$, the inefficiency of ML-WCP decreases significantly. For example, with $\alpha=0.1$, the average inefficiency of ML-WCP is $1.26$, while for CCP one obtains $1.54$.

\subsection{Profiling Medium Access Control Scheduling Apps}
Consider now the MAC layer problem described in Section \ref{ex_mac scheduling}, in which the objective is to calibrate a pre-trained predictive model to estimate the action of a third-party scheduling AI app. The 5G-based simulation environment is configured as in \cite{alvarovalcarce2020Nokia}. Specifically, the number of downlink resource blocks is set to $25$, and each episode begins with $|\mathcal{Y}|=32$ randomly distributed UEs. The buffer size of each UE is limited to $8$ packets. The episode is divided into $10,000$ transmission time intervals (TTIs), and we monitor the scheduling actions every $250$ TTIs. At the beginning of each TTI, new packets are randomly generated for each UE and allocated to its buffer for scheduling. The scheduling app is a third-party AI-driven system designed to allocate downlink resource blocks to UEs based on their backlog and CQI \cite{Ana2020RLScheduler}.

The predictive model takes as input $x$ the current backlogs and the corresponding CQIs while the target output $y \in \mathcal{Y}$ represents the action of selecting which UE to schedule. The predictive model is implemented with a fully connected network with two hidden layers of $256$ neurons with ReLU activation.

The context $c$ encapsulates the information about the backlogs and CQIs at the previous scheduling interval. In this way, $|\mathcal{D}_c|=150$ data points are generated for each of $|\mathcal{C}| = 40$ contexts. Out of these, up to $|\mathcal{C}^\text{tr}| = 20$ contexts are allocated for meta-learning, with the remaining contexts are reserved for evaluating the calibration performance.

\begin{figure*}[t]
   \centering
   \hspace*{+0.75cm}
    \includegraphics[width=1.05\linewidth]{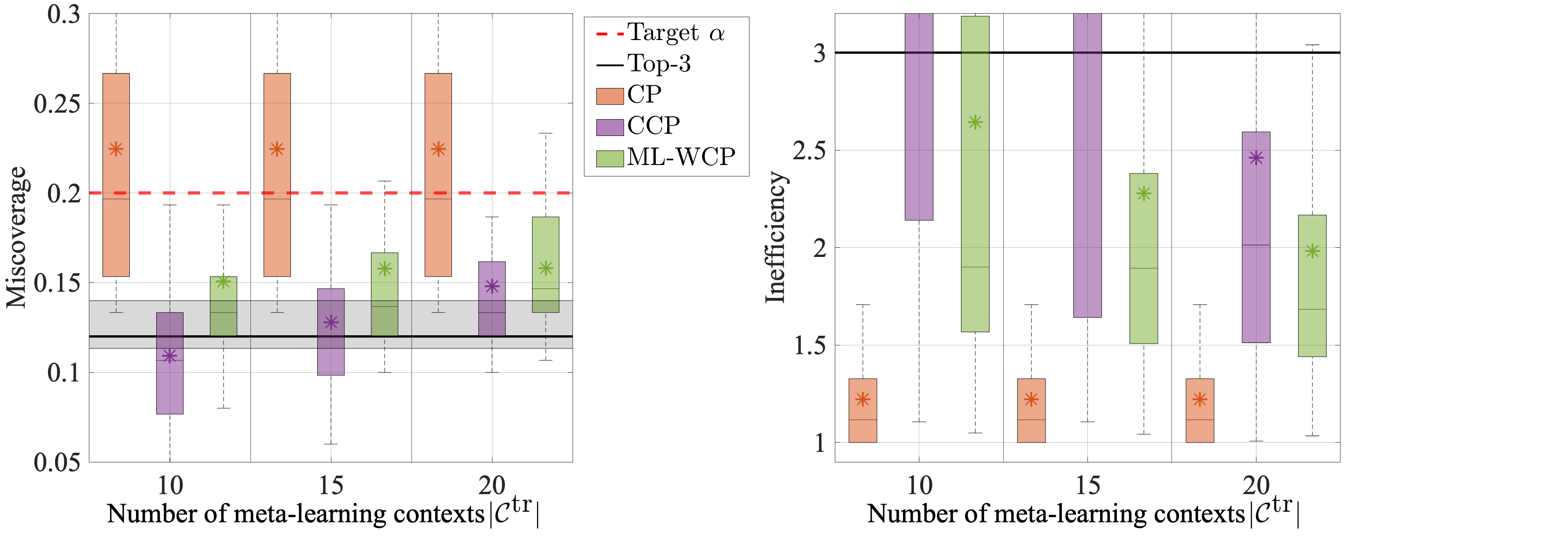}
    \caption{Empirical coverage and inefficiency of Top-3 prediction, CP, CCP, and ML-WCP versus the number of meta-learning contexts with target miscoverage level $\alpha=0.2$.}
    \label{app2_metatrain}
\end{figure*}

In Fig. \ref{app2_metatrain}, we evaluate the performance of ML-WCP against the benchmarks given by Top-3 prediction, CP, and CCP as a function of the number of meta-learning contexts $|\mathcal{C}^\text{tr}|$. We set the target miscoverage level as $\alpha=0.2$, observing that the Top-3 prediction method ensures sufficient coverage with this choice. The results reported in the figure confirm that, owing to covariate shift, the CP method fails to provide the coverage guarantee (\ref{cp_coverage}). In contrast, the alternative methods consistently maintain the miscoverage level below the target $\alpha$. For both CCP and ML-WCP, the miscoverage level approaches the target level $\alpha$ as the number of meta-learned contexts increases, while the inefficiency decreases. However, CCP shows higher variability in the inefficiency, possibly returning uninformative prediction sets. In contrast, ML-WCP consistently achieves reliable coverage while maintaining an average inefficiency level below $3$.

\begin{figure}
    \centering
    \hspace*{+0.75cm}
    \includegraphics[width=0.6\linewidth]{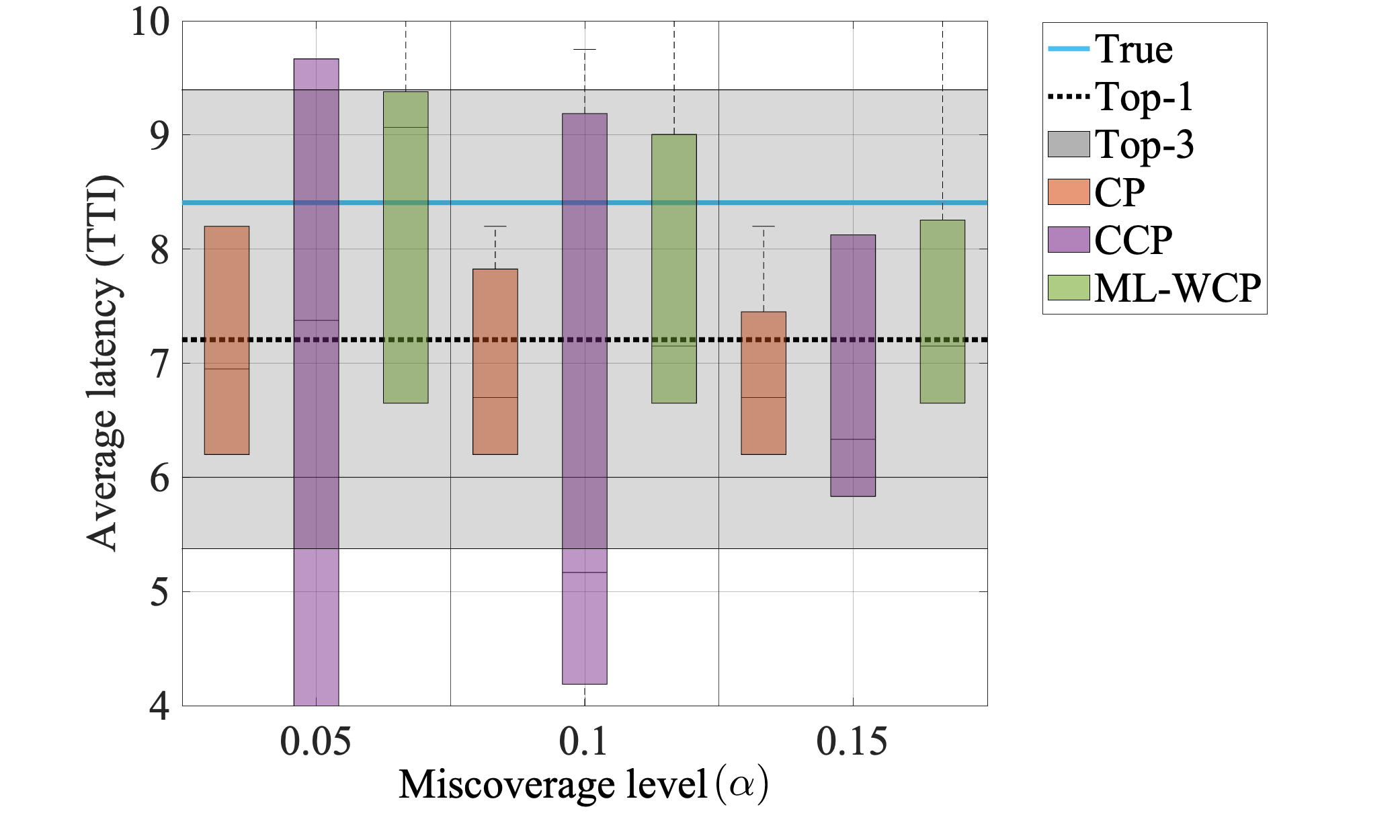}
    \caption{Average latency predicted by Top-1 and Top-3 methods, as well as by CP, CCP, and ML-WCP, versus the miscoverage level $\alpha$. The solid line represents the true average latency.}
    \label{app2_alpha_petar}
\end{figure}

As an application of prediction sets, we consider now the problem of estimating the KPI levels attained by the third-party scheduling app using the pre-trained model. We specifically focus on the average latency, which is defined as the average number of TTIs required to deliver a packet intended for any of the UEs. The app is executed over 10 TTIs. As shown in Fig. \ref{app2_alpha_petar}, we compare the performance of Top-1 and Top-3 prediction, as well as of CP, CCP, and ML-WCP, against the true values calculated by running the actual app. For Top-3 prediction, CP, CCP, and ML-WCP, the estimate is obtained by sampling scheduling actions uniformly from the prediction set. The performance is averaged over $30$ test instances, and by producing $5$ random estimates for each test instance. The figure shows the band of estimated values for each scheme.

By construction, Top-3 prediction obtains estimation intervals that do not depend on the miscoverage level $\alpha$. In contrast, the CP, CCP, and ML-WCP schemes exhibit decreasing interval sizes as the miscoverage rate $\alpha$ increases due to the reduction in the size of the prediction sets. Top-1 and CP always underestimate the latency. Furthermore, both CCP and ML-WCP produce intervals containing the true average latency as long as the miscoverage rate is set as $\alpha \leq 0.1$. However, CCP produces wide estimation intervals, while ML-WCP yields significantly small intervals.

\begin{figure*}
    \centering
    \hspace*{-0.75cm}\includegraphics[width=1\linewidth]{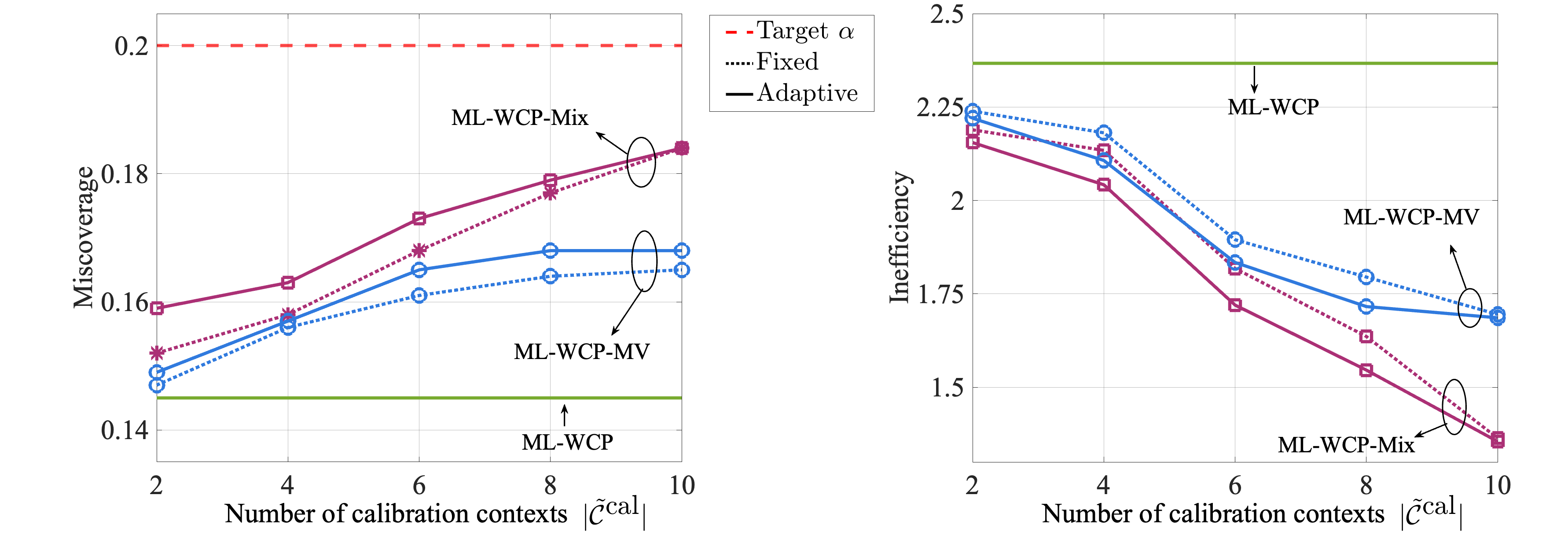}
    \caption{Empirical coverage and inefficiency of ML-WCP, ML-WCP-MV, and ML-WCP-Mix as a function of the average number of calibration contexts $|\tilde{\mathcal{C}}^\text{cal}|$ under different context selection methods, with a target miscoverage level of $\alpha=0.2$.}
    \label{app2_multicontext}
\end{figure*}

In Fig. \ref{app2_multicontext}, we present the performance of the \textit{multi-context} ML-WCP schemes as a function of the number of selected calibration contexts $|\tilde{\mathcal{C}}^\text{cal}|$, we have a total number $|\mathcal{C}^\text{cal}|=10$ of calibration contexts and $|\mathcal{C}^\text{tr}|=20$ meta-learning contexts. We consider fixed selection schemes that always choose $|\tilde{\mathcal{C}}^\text{cal}|$ calibration contexts using (\ref{multi_cal_select_fixed}), as well as adaptive schemes that use $|\tilde{\mathcal{C}}^\text{cal}|$ calibration contexts on average by using the selection strategy (\ref{multi_cal_select_adapt}). In practice, adaptive methods vary the average number of selected calibration contexts by modifying the threshold $\epsilon$ in (\ref{multi_cal_select_adapt}).

The figure reveals that adaptive selection, which flexibly adjusts the number of calibration contexts for each test context, consistently achieves lower inefficiency than methods based on a fixed number of calibration contexts. Compared to ML-WCP, which considers a single calibration context, the multi-context methods ML-WCP-MV and ML-WCP-Mix achieve a lower inefficiency while maintaining coverage guarantees. Notably, the performance gap becomes more pronounced as the number of calibration contexts $|\tilde{\mathcal{C}}^\text{cal}|$ increases. In addition, ML-WCP-Mix outperforms ML-WCP-MV regardless of the multi-context selection method. This confirms the benefits of jointly leveraging multiple calibration contexts via the mixture distribution (\ref{mixture_dist}).

\subsection{Interference-Limited Phyiscal-Layer Communication} \label{Interference-Limited Communication}

As a first application, we consider the interference-limited communication scenario described in Section \ref{ex_interference_limited_communication}. In this problem, as depicted in Fig. \ref{application}(c), the received
signal $x$ serves as the input, while the corresponding sequence of encoded information bits corresponds to the target variables $y$. The decoder serves as the underlying model that assigns scores $S(x,y)$ to the possible information messages $y$ given the received signal $x$ \cite{tse2005fundamentals}. In this experiment, the Bluetooth Low Energy (BLE) 5 standard serves as a framework for validating the described scenario \cite{heydon2012bluetooth,Morais2023bluetooth}. Accordingly, we consider short blocks of $8$ information bits encoded by a convolutional encoder and modulated with 4-QAM, resulting in $8$ symbols. 

These symbols pass through an AWGN channel, becoming the input to the decoder, which selects a prediction set from the $|\mathcal{Y}|=256$ possible information messages of $8$ bits. Thus, the prediction set generated by the pre-trained predictive model can be viewed as the output of a list decoder, designed to include a set of candidate values \cite{correia2024cp_listdecoding}. Miscoverage corresponds to the list decoding error, defined as instances where the true value is absent from the selected set. Inefficiency is the average list size of the selected candidate messages. 

In the presence of interference, which is indicated as $I_b=1$, burst noise is modeled as AWGN with a higher noise power. The starting time $T_0$ of the interference and the duration $T_b$ follows a uniform distribution. Specifically when interference is present, the signal-to-interference-plus-noise ratio (SINR),
\begin{equation}
    \textrm{SINR}=\frac{\textrm{SNR}}{1+\textrm{INR}},
\end{equation}
is smaller than the SNR level by a multiplicative factor that depends on the interference-to-noise ratio (INR). We primarily set $\textrm{SNR}=1$ dB and $\textrm{INR}=-7.5$ dB, yielding $\text{SINR}=0.3$ dB.

To explore the impact of the informativeness of the context vector $c$, we focus on the following options.
\begin{itemize}
    \item \textit{Least informative context}: The context variable $c=(p_b,T_b)$ includes the probability $p_b$ of occurrence of an interference burst in a frame and the duration $T_b$ of the burst.
    \item \textit{Moderately informative context}: The context variable $c=(I_b,T_b)$ includes a binary indicator $I_b \in \{0,1\}$ for whether interference occurs, $I_b=1$ or not, $I_b=0$, along with the duration $T_b$ of the interference burst.
    \item \textit{Most informative context}: The context variable $c=(I_b,T_b,T_0)$ encompasses also the starting time $T_0$ of the interference burst.
\end{itemize}

Note that, in this simple example, the exact covariate likelihood ratio can be derived as described in Appendix B, making it possible to compare the performance of all the schemes against Ideal WCP.

\begin{figure*}[t]
   \centering
   \hspace*{+0.75cm}   
    \includegraphics[width=1.05\linewidth]{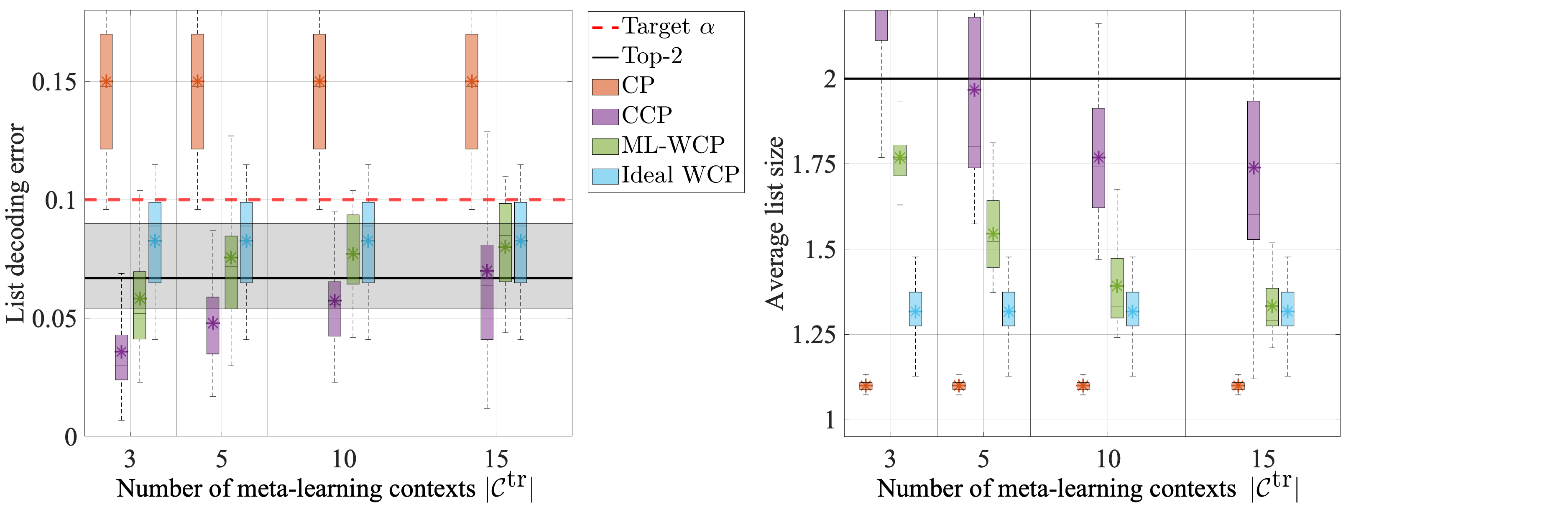}
    \caption{Empirical list decoding error and average list size of Top-2 prediction, CP, CCP, ML-WCP, and Ideal WCP versus the number of meta-learning contexts with most informative context vector, target miscoverage level $\alpha=0.1$.}
    \label{app3_numberofcontexts}
\end{figure*}
In Fig. \ref{app3_numberofcontexts}, we demonstrate the performance of the considered set predictors as a function of the number of meta-learning contexts $|\mathcal{C}^\text{tr}|$ assuming the availability of the most informative context. The figure reveals that CP fails to meet the target miscoverage $\alpha$, yielding an average list decoding error of $0.151$. The Top-2 and CCP methods demonstrate a lower list decoding error than the target $\alpha$, but at the cost of a higher average list size.  In contrast, ML-WCP gradually aligns with the performance of Ideal WCP as the number of meta-learned contexts increases.

\begin{figure*}[t]
   \centering 
   \hspace*{+0.75cm}
    \includegraphics[width=1.05\linewidth]{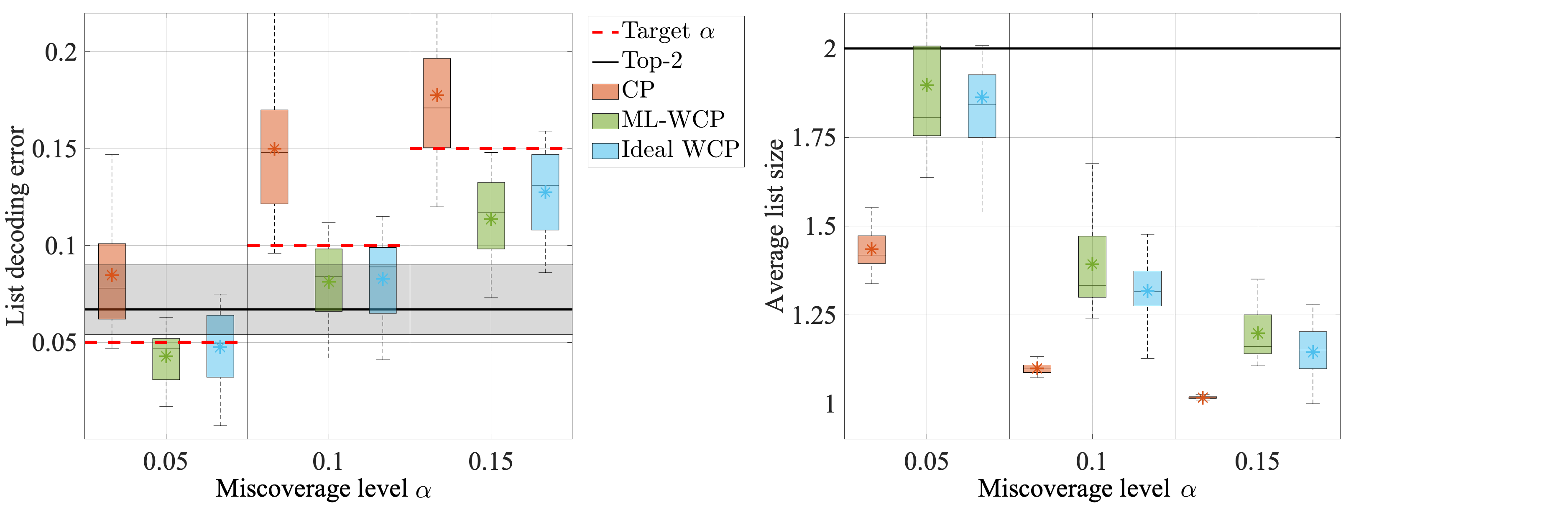}
    \caption{Empirical list decoding error and average list size of Top-2 prediction, CP, ML-WCP, and Ideal WCP versus the miscoverage level $\alpha$ with most informative context case, the number of meta-learning contexts $|\mathcal{C}^\text{tr}|=10$.}
    \label{app3_alpha}
\end{figure*}

In Fig. \ref{app3_alpha}, we evaluate the empirical list decoding error and average list size with respect to the target miscoverage level $\alpha$ with $|\mathcal{C}^\text{tr}|=10$ meta-learning contexts and with the most informative context. For the Top-$2$ prediction method, the list decoding error remains negligible across all values of $\alpha$, and the list size is consistently fixed at $2$. In contrast, for other methods, an increase in the target miscoverage level $\alpha$ results in smaller prediction sets. Unlike standard CP, which fails to meet the target reliability level, the proposed ML-WCP achieves the coverage, closely matching the average list size of the Ideal WCP. For instance, at $\alpha=0.15$, the average list size of the ML-WCP is $1.19$, while that of Ideal WCP is $1.15$, demonstrating a minimal gap.

\begin{figure*}[t]
   \centering
   \hspace*{+0.01cm}
    \includegraphics[width=1.05\linewidth]{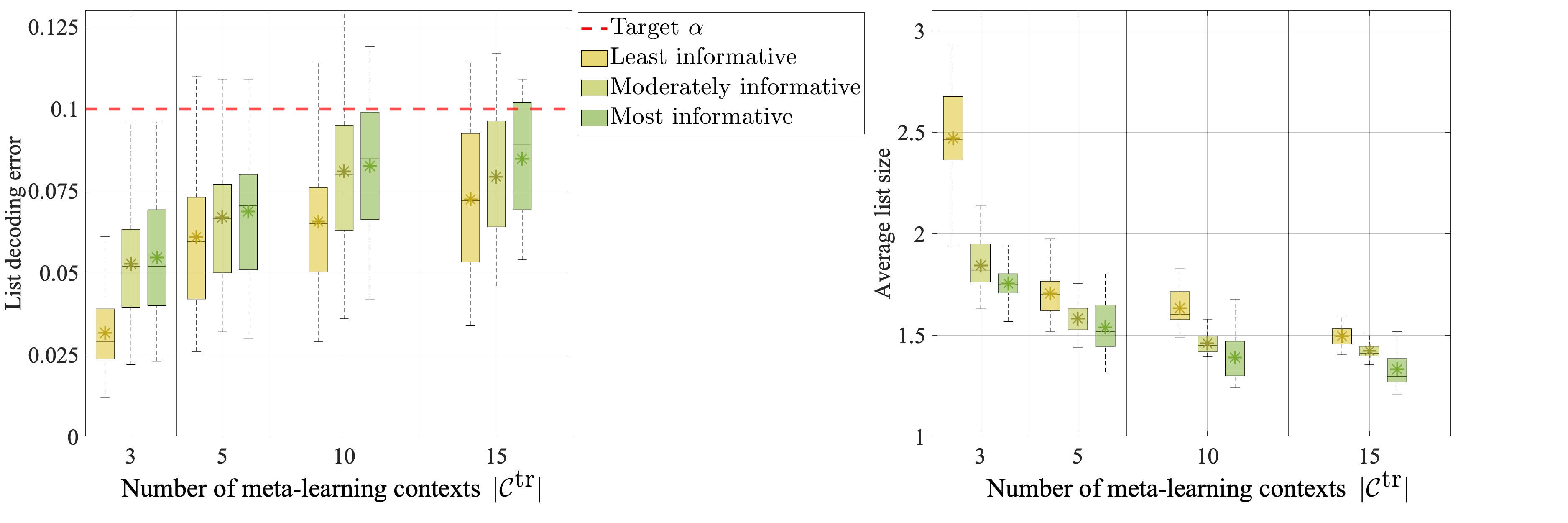}
    \caption{Empirical list decoding error and average list size of ML-WCP for different levels of informativeness of the context vector versus the number of meta-learning contexts with target miscoverage level $\alpha=0.1$.}
    \label{app3_informative}
\end{figure*}

Furthermore, we investigate the impact of the informativeness of the context vector as a function of number of meta-learning contexts $|\mathcal{C}^\text{tr}|$. As shown in Fig. \ref{app3_informative}, when the number of meta-learning contexts is limited, i.e., $|\mathcal{C}^\text{tr}|=3$, there are notable performance gaps in list decoding error and average list size performance depending on the level of context informativeness. Specifically, incorporating detailed context features, such as interference burst occurrence and start time, proves crucial for improving performance under constrained meta-learning settings. As the number of meta-learning contexts increases, when $|\mathcal{C}^\text{tr}|=15$, the average list size stabilizes under $1.5$, demonstrating the stability and robustness of ML-WCP across varying levels of contextual information.

\begin{figure*}[t]
   \centering
   \hspace*{+0.5cm}
    \includegraphics[width=1.05\linewidth]{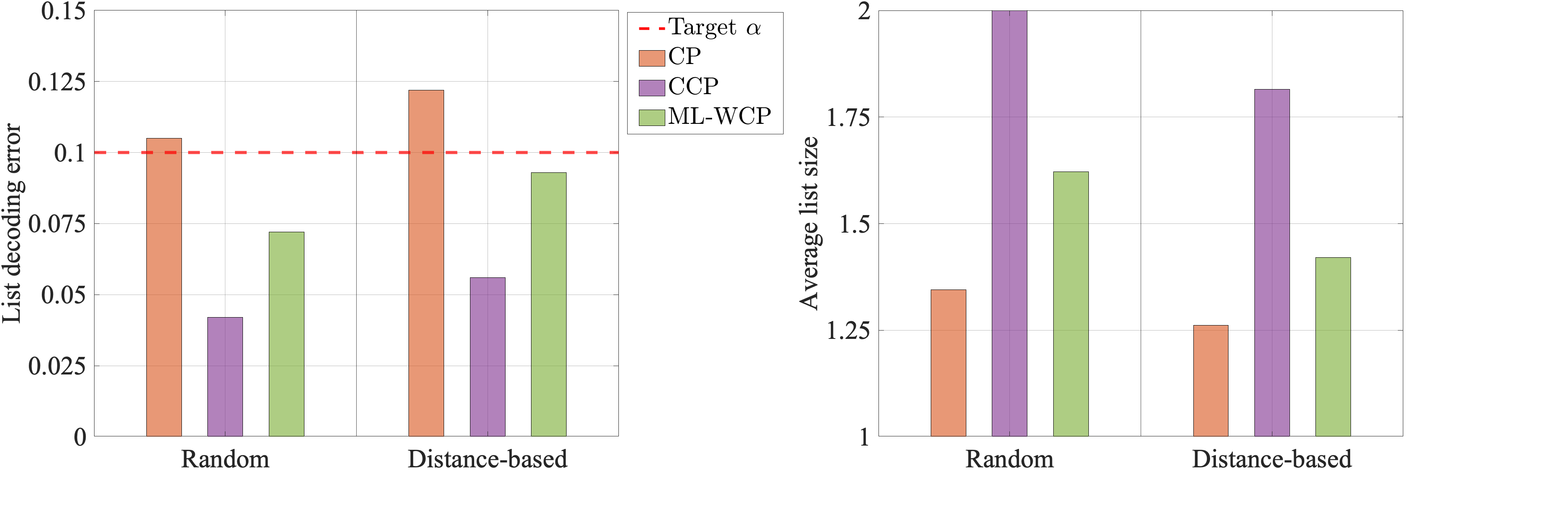}
    \caption{Empirical list decoding error and average list size of CP, CCP, and ML-WCP for different context selection methods with target miscoverage level $\alpha=0.1$.}
    \label{app3_context_select}
\end{figure*}

Fig. \ref{app3_context_select} evaluates the runtime calibration performance when applying different context selection methods. Specifically, we compare the random selection method adopted so far with the distance-based selection metric presented in Section \ref{calibration data}. Across all methods, distance-based selection is seen to lead to reduced prediction set sizes compared to random selection. Notably, CCP shows significant performance variability depending on the context selection method. For the random selection method, CCP achieves the list decoding error below $0.05$ but results in significantly large prediction set sizes. In contrast, the proposed ML-WCP achieves target coverage reliability condition (\ref{wcp_coverage}) when using distance-based selection, closely approaching the actual minimum average list size of $1.42$.

% \begin{figure}[h]
%      \centering
%      \begin{subfigure}[t]{0.4\textwidth}
%          \centering
%          \includegraphics[width=\textwidth]{dTV-acc.png}
%          \label{dTV_acc}
%      \end{subfigure}
%      \hspace{-10pt}
%      \begin{subfigure}[t]{0.4\textwidth}
%          \centering
%          \includegraphics[width=\textwidth]{dTV-size.png}
%          \label{dTV_size}
%      \end{subfigure}
%         \caption{.}
%         \label{d_TV}
% \end{figure}
\section{Conclusion} \label{Conclusion}
Pre-deployment calibration is instrumental in facilitating the adoption of AI models in wireless systems, which are characterized by dynamic and diverse network conditions. Calibration can ensure performance guarantees via single post-hoc mechanisms that augment decisions with prediction sets. Existing state-of-the-art calibration methods based on the WCP framework require knowledge of the distribution shift between calibration and runtime operation. However, this is not available in practical wireless settings. The proposed ML-WCP leverages contextual information to estimate the distribution shift without requiring runtime data, and it can be extended to integrate multi-context information by leveraging calibration data from multiple contexts via meta-learning. To validate the reliability and performance of ML-WCP, we demonstrated its application across three wireless scenarios operating at the network layer, MAC layer, and physical layer.

Future research may focus on enhancing ML-WCP by supporting the online optimization of scoring functions, or providing deterministic coverage guarantees under adversarial network conditions.

\appendix
\renewcommand{\qedsymbol}{$\square$}
\subsection{Proof of Lemma 2}

Let the estimated target distribution ${p}_{\theta}(x|c^\text{te})$ be obtained for the estimated covariate likelihood ratio $\omega(x,c^\text{te},c^\text{cal})$ as
\begin{equation}
{p}_{\theta}(x|c^\text{te})=\frac{\omega_{\theta}(x,c^\text{te},c^\text{cal})}{\mathbb{E}_{x \sim p(x|c^\textnormal{cal})}[\omega_{\theta}(x,c^\textnormal{te},c^\textnormal{cal})]}\cdot p(x|c^\text{cal}).
\end{equation}
Due to the gap between the coverage under the estimated target distribution ${p}_{\theta}(x|c^\text{te})$ and the true target distribution $p(x|c^\text{te})$, which is quantified by the total-variation distance, the coverage of ML-WCP satisfies the inequality
\begin{equation}
\begin{aligned}
    \text{Pr}\big[y^\text{te} \in \Gamma_{\alpha}^\text{ML-WCP}(x^\text{te}|c^\textnormal{te},c^\textnormal{cal})\big]&\geq 1-\alpha-\frac{1}{2}\int\big|{p}_{\theta}(x|c^\text{te})-p(x|c^\text{te})\big|\mathrm{d}x\\
    &=1-\alpha-\frac{1}{2}\int\bigg|\frac{\omega_{\theta}(x,c^\textnormal{te},c^\textnormal{cal})}{\mathbb{E}_{x \sim p(x|c^\textnormal{cal})}[\omega_{\theta}(x,c^\textnormal{te},c^\textnormal{cal})]}-w(x,c^\text{te},c^\text{cal})\bigg|p(x|c^\text{cal})\mathrm{d}x\\
    &=1-\alpha-\frac{1}{2}\mathbb{E}_{x\sim p(x|c^\text{cal})}\bigg|\frac{w_{\theta}(x,c^\textnormal{te},c^\textnormal{cal})}{\mathbb{E}_{x \sim p(x|c^\textnormal{cal})}[\omega_{\theta}(x,c^\textnormal{te},c^\textnormal{cal})]}-w(x,c^\textnormal{te},c^\textnormal{cal})\bigg|.
\end{aligned}
\end{equation}
% \subsection{Proof of Lemma 3}
% \begin{proof}
% Using the definition (\ref{multi-context likelihood}), we write
% \begin{equation}
% \begin{aligned}
%     w(x, c_1, \mathcal{C}_2) &= \frac{p(x|c_1)}{\frac{1}{|\mathcal{C}_2|} \sum_{c \in \mathcal{C}_2} p(x|c)}\\
%     &=\frac{{q(c_1|x)}}{\frac{1}{|\mathcal{C}_2|} \sum_{c \in \mathcal{C}_2} {q(c|x)}}\frac{\big(p(x|c_1)+\frac{1}{|\mathcal{C}_2|}\sum_{c \in \mathcal{C}_2}p(x|c)\big)}{\big(p(x|c_1)+\frac{1}{|\mathcal{C}_2|}\sum_{c \in \mathcal{C}_2}p(x|c)\big)}\\
%     &=\frac{q(c_1|x)}{\frac{1}{|\mathcal{C}_2|} \sum_{c \in \mathcal{C}_2} {q(c|x)}}.
% \end{aligned}
% \end{equation}
% \end{proof}

\subsection{Conditional Distribution in Section \ref{Interference-Limited Communication}}
In this subsection, we specify the details for deriving the conditional distribution $p(x|c)$ in the interference-limited wireless communication scenario studied in Section \ref{Interference-Limited Communication}. We focus on the setting with 4-QAM modulation. The received complex signal $x$ consists of $T$ elements, and each $t$-th element is denoted by $x_t$. The $t$-th transmitted signal $s_t$, is selected from the set $\mathcal{S}\triangleq \{(1+j,1-j,-1+j,-1-j)/\sqrt{2}\}$. The conditional distribution $p(x|c)$ is then given by
\begin{equation}
\begin{aligned}
   p(x|c)&=\prod\limits_{t=1}^{T}p(x_t|c)
\\ &= \sum_{s_t \in \mathcal{S}}p(x_t|s_t,c)p(s_t|c). 
\end{aligned}
\end{equation}
Assuming that $s_t$ is independent of context $c$, we write $p(s_t|c)= p(s_t)=1/4$ for each symbol. Further, under the assumption that the context $c$ fully represents interference information, $p(x_t|s_t,c)$ is specified by
\begin{equation}
    p(x_t|s_t,c) = \mathbbm{1}(t\in[T_0,T_b+T_0])\mathcal{N}(s_t,\sigma_0^2+\sigma_1^2)+\mathbbm{1}(t\notin[T_0,T_b+T_0])\mathcal{N}(s_t,\sigma_0^2),
\end{equation}
where $\sigma_0^2$ and $\sigma_1^2$ represent the variances of noise and interference, respectively. Similar derivations apply to all the cases studied in Section \ref{ex_wireless communication}.

\bibliographystyle{IEEEtran}
\bibliography{ref}

\begin{thebibliography}{10}
\providecommand{\url}[1]{#1}
\csname url@samestyle\endcsname
\providecommand{\newblock}{\relax}
\providecommand{\bibinfo}[2]{#2}
\providecommand{\BIBentrySTDinterwordspacing}{\spaceskip=0pt\relax}
\providecommand{\BIBentryALTinterwordstretchfactor}{4}
\providecommand{\BIBentryALTinterwordspacing}{\spaceskip=\fontdimen2\font plus
\BIBentryALTinterwordstretchfactor\fontdimen3\font minus \fontdimen4\font\relax}
\providecommand{\BIBforeignlanguage}[2]{{%
\expandafter\ifx\csname l@#1\endcsname\relax
\typeout{** WARNING: IEEEtran.bst: No hyphenation pattern has been}%
\typeout{** loaded for the language `#1'. Using the pattern for}%
\typeout{** the default language instead.}%
\else
\language=\csname l@#1\endcsname
\fi
#2}}
\providecommand{\BIBdecl}{\relax}
\BIBdecl

\bibitem{Abdalla2022oran}
A.~S. Abdalla, P.~S. Upadhyaya, V.~K. Shah, and V.~Marojevic, ``Toward next generation open radio access networks: What {O-RAN} can and cannot do!'' \emph{IEEE Netw.}, vol.~36, no.~6, pp. 206--213, Nov. 2022.

\bibitem{Latah2019ai_apps}
M.~Latah and L.~Toker, ``Artificial intelligence enabled software-defined networking: A comprehensive overview,'' \emph{IET Netw.}, vol.~8, no.~2, pp. 79--99, 2019.

\bibitem{Bonati2021phy_functions}
L.~Bonati, S.~D'Oro, M.~Polese, S.~Basagni, and T.~Melodia, ``Intelligence and learning in {O-RAN} for data-driven {N}ext{G} cellular networks,'' \emph{IEEE Commun. Mag.}, vol.~59, no.~10, pp. 21--27, Oct. 2021.

\bibitem{qazzaz2024higher_layer}
M.~M.~H. Qazzaz, Łukasz Kułacz, A.~Kliks, S.~A. Zaidi, M.~Dryjanski, and D.~McLernon, ``Machine learning-based x{A}pp for dynamic resource allocation in {O-RAN} networks,'' \emph{arXiv:2401.07643}, 2024.

\bibitem{Brik2024oran_automation}
B.~Brik, H.~Chergui, L.~Zanzi, F.~Devoti, A.~Ksentini, M.~S. Siddiqui, X.~Costa-Pèrez, and C.~Verikoukis, ``Explainable {AI} in 6{G} {O-RAN}: {A} tutorial and survey on architecture, use cases, challenges, and future research,'' \emph{IEEE Commun. Surv. Tutor.}, pp. 1--1, 2024.

\bibitem{Hamdan2023ORAN_automation}
M.~Q. Hamdan, H.~Lee, D.~Triantafyllopoulou, R.~Borralho, A.~Kose, E.~Amiri, D.~Mulvey, W.~Yu, R.~Zitouni, R.~Pozza, B.~Hunt, H.~Bagheri, C.~H. Foh, F.~Heliot, G.~Chen, P.~Xiao, N.~Wang, and R.~Tafazolli, ``Recent advances in machine learning for network automation in the {O-RAN},'' \emph{Sensors}, vol.~23, no.~21, 2023.

\bibitem{guo2017calibration}
C.~Guo, G.~Pleiss, Y.~Sun, and K.~Q. Weinberger, ``On calibration of modern neural networks,'' in \emph{International conference on machine learning}.\hskip 1em plus 0.5em minus 0.4em\relax PMLR, 2017, pp. 1321--1330.

\bibitem{kuleshov2018accurate}
V.~Kuleshov, N.~Fenner, and S.~Ermon, ``Accurate uncertainties for deep learning using calibrated regression,'' in \emph{International conference on machine learning}.\hskip 1em plus 0.5em minus 0.4em\relax PMLR, 2018, pp. 2796--2804.

\bibitem{marx2022modular}
C.~Marx, S.~Zhao, W.~Neiswanger, and S.~Ermon, ``Modular conformal calibration,'' in \emph{International Conference on Machine Learning}.\hskip 1em plus 0.5em minus 0.4em\relax PMLR, 2022, pp. 15\,180--15\,195.

\bibitem{angelopoulos2024therory_CP}
A.~N. Angelopoulos, R.~F. Barber, and S.~Bates, ``Theoretical foundations of conformal prediction,'' \emph{arXiv:2411.11824}, 2024.

\bibitem{shafer2008CP}
G.~Shafer and V.~Vovk, ``A tutorial on conformal prediction,'' \emph{J. Mach. Learn. Res.}, vol.~9, no.~3, Mar. 2008.

\bibitem{balasubramanian2014cp_book}
V.~Balasubramanian, S.-S. Ho, and V.~Vovk, \emph{Conformal prediction for reliable machine learning: theory, adaptations and applications}.\hskip 1em plus 0.5em minus 0.4em\relax Newnes, 2014.

\bibitem{YANG2020hyperparam_review}
L.~Yang and A.~Shami, ``On hyperparameter optimization of machine learning algorithms: Theory and practice,'' \emph{Neurocomputing}, vol. 415, pp. 295--316, Nov. 2020.

\bibitem{Angelopoulos2020Top-k}
A.~Angelopoulos, S.~Bates, J.~Malik, and M.~I. Jordan, ``Uncertainty sets for image classifiers using conformal prediction,'' \emph{arXiv:2009.14193}, 2020.

\bibitem{Minderer2021NEURIPS}
M.~Minderer, J.~Djolonga, R.~Romijnders, F.~Hubis, X.~Zhai, N.~Houlsby, D.~Tran, and M.~Lucic, ``Revisiting the calibration of modern neural networks,'' in \emph{Proc. Adv. Neural Inf. Process. Syst. (NIPs)}, vol.~34, 2021, pp. 15\,682--15\,694.

\bibitem{wang2024ml_calibration}
C.~Wang, ``Calibration in deep learning: A survey of the state-of-the-art,'' \emph{arXiv:2308.01222}, 2023.

\bibitem{Tibshirani2019Dist_shift}
R.~J. Tibshirani, R.~Foygel~Barber, E.~Candes, and A.~Ramdas, ``Conformal prediction under covariate shift,'' in \emph{Proc. Adv. Neural Inf. Process. Syst. (NIPs)}, vol.~32, 2019, pp. 2526--2536.

\bibitem{Lei2021TV}
L.~Lei and E.~J. Candès, ``{Conformal inference of counterfactuals and individual treatment effects},'' \emph{J. R. Stat. Soc. Ser. B. Stat. Methodol.}, vol.~83, no.~5, pp. 911--938, 10 2021.

\bibitem{Yang2024DoublyRobustCP}
Y.~Yang, A.~K. Kuchibhotla, and E.~Tchetgen~Tchetgen, ``{Doubly robust calibration of prediction sets under covariate shift},'' \emph{J. R. Stat. Soc. B}, 03 2024.

\bibitem{prinster2024conformal}
D.~Prinster, S.~Stanton, A.~Liu, and S.~Saria, ``Conformal validity guarantees exist for any data distribution,'' \emph{arXiv preprint arXiv:2405.06627}, 2024.

\bibitem{barber2023conformal}
R.~F. Barber, E.~J. Candes, A.~Ramdas, and R.~J. Tibshirani, ``Conformal prediction beyond exchangeability,'' \emph{The Annals of Statistics}, vol.~51, no.~2, pp. 816--845, 2023.

\bibitem{cauchois2024robust}
M.~Cauchois, S.~Gupta, A.~Ali, and J.~C. Duchi, ``Robust validation: Confident predictions even when distributions shift,'' \emph{Journal of the American Statistical Association}, pp. 1--66, 2024.

\bibitem{hou2024likelihood}
Q.~Hou, S.~Park, M.~Zecchin, Y.~Cai, G.~Yu, and O.~Simeone, ``What if we had used a different app? {R}eliable counterfactual {KPI} analysis in wireless systems,'' \emph{arxiv:2410.00150}, 2024.

\bibitem{Sangwoo2023meta-learning}
L.~Chen, S.~T. Jose, I.~Nikoloska, S.~Park, T.~Chen, and O.~Simeone, ``Learning with limited samples: Meta-learning and applications to communication systems,'' \emph{Found. Trends Signal Process.}, vol.~17, no.~2, pp. 79--208, 2023.

\bibitem{Wang2024context_controller}
S.~Wang, M.~Ruiz, and L.~Velasco, ``Context-based e2e autonomous operation in {B5G} networks,'' \emph{Sensors}, vol.~24, no.~5, 2024.

\bibitem{Semiari2015Context}
O.~Semiari, W.~Saad, S.~Valentin, M.~Bennis, and H.~V. Poor, ``Context-aware small cell networks: How social metrics improve wireless resource allocation,'' \emph{IEEE Trans. Wireless Commun.}, vol.~14, no.~11, pp. 5927--5940, Nov. 2015.

\bibitem{Letaief2019AI}
K.~B. Letaief, W.~Chen, Y.~Shi, J.~Zhang, and Y.-J.~A. Zhang, ``The roadmap to 6{G}: {AI} empowered wireless networks,'' \emph{IEEE Commun. Mag.}, vol.~57, no.~8, pp. 84--90, 2019.

\bibitem{Zappone2019AI}
A.~Zappone, M.~Di~Renzo, and M.~Debbah, ``Wireless networks design in the era of deep learning: {M}odel-based, {AI}-based, or both?'' \emph{IEEE Trans. Commun.}, vol.~67, no.~10, pp. 7331--7376, 2019.

\bibitem{guo2020explainable}
W.~Guo, ``Explainable artificial intelligence for 6{G}: {I}mproving trust between human and machine,'' \emph{IEEE Commun. Mag.}, vol.~58, no.~6, pp. 39--45, 2020.

\bibitem{simeone2022machine}
O.~Simeone, \emph{Machine Learning for Engineers}.\hskip 1em plus 0.5em minus 0.4em\relax Cambridge, U.K.: Cambridge University Press, 2022.

\bibitem{cohen2021learning}
K.~M. Cohen, S.~Park, O.~Simeone, and S.~Shamai, ``Learning to learn to demodulate with uncertainty quantification via bayesian meta-learning,'' in \emph{WSA 2021; 25th International ITG Workshop on Smart Antennas}, 2021, pp. 1--6.

\bibitem{raviv2023modular}
T.~Raviv, S.~Park, O.~Simeone, and N.~Shlezinger, ``Modular model-based bayesian learning for uncertainty-aware and reliable deep {MIMO} receivers,'' in \emph{IEEE ICC Workshops}, 2023, pp. 1032--1037.

\bibitem{tedeschini2024real}
B.~C. Tedeschini, G.~Kwon, M.~Nicoli, and M.~Z. Win, ``Real-time bayesian neural networks for 6{G} cooperative positioning and tracking,'' \emph{IEEE J. Sel. Areas Commun.}, vol.~42, no.~9, pp. 2322--2338, 2024.

\bibitem{knoblauch2019generalized}
J.~Knoblauch, J.~Jewson, and T.~Damoulas, ``Generalized variational inference: Three arguments for deriving new posteriors,'' \emph{arXiv:1904.02063}, 2019.

\bibitem{Zecchin2023Bayesian}
M.~Zecchin, S.~Park, O.~Simeone, M.~Kountouris, and D.~Gesbert, ``Robust bayesian learning for reliable wireless {AI}: {F}ramework and applications,'' \emph{IEEE Trans. Cogn. Commun. Netw.}, vol.~9, no.~4, pp. 897--912, 2023.

\bibitem{vovk2005cp}
V.~Vovk, A.~Gammerman, and G.~Shafer, \emph{Algorithmic learning in a random world}.\hskip 1em plus 0.5em minus 0.4em\relax Berlin, Germany: Springer, 2005, vol.~29.

\bibitem{Cohen2023WirelessCP}
K.~M. Cohen, S.~Park, O.~Simeone, and S.~Shamai~Shitz, ``Calibrating ai models for wireless communications via conformal prediction,'' \emph{IEEE Trans. Mach. Learn. Commun. Netw.}, vol.~1, pp. 296--312, 2023.

\bibitem{jiang2024learning}
H.~Jiang, E.~Belding, E.~Zegure, and Y.~Xie, ``Learning cellular network connection quality with conformal,'' \emph{arXiv:2407.10976}, 2024.

\bibitem{ma2023metastnet}
H.~Ma and K.~Yang, ``Metastnet: Multimodal meta-learning for cellular traffic conformal prediction,'' \emph{IEEE Trans. Netw. Sci. Eng.}, 2023.

\bibitem{gibbs2021adaptive}
I.~Gibbs and E.~Candes, ``Adaptive conformal inference under distribution shift,'' \emph{Adv. Neural Inf. Process. Syst.}, vol.~34, pp. 1660--1672, 2021.

\bibitem{feldman2022achieving}
S.~Feldman, L.~Ringel, S.~Bates, and Y.~Romano, ``Achieving risk control in online learning settings,'' \emph{arXiv:2205.09095}, 2022.

\bibitem{Cohen2023URLLC_CP}
K.~M. Cohen, S.~Park, O.~Simeone, P.~Popovski, and S.~Shamai, ``Guaranteed dynamic scheduling of ultra-reliable low-latency traffic via conformal prediction,'' \emph{IEEE Signal Process. Lett.}, vol.~30, pp. 473--477, 2023.

\bibitem{Anastasios2023CPintro}
A.~N. Angelopoulos and S.~Bates, ``Conformal prediction: A gentle introduction,'' \emph{Found. Trends Mach. Learn.}, vol.~16, no.~4, pp. 494--591, 2023.

\bibitem{Chowdbury[2024]TrafficPred.}
J.~Groen, M.~Belgiovine, U.~Demir, B.~Kim, and K.~Chowdhury, ``{TRACTOR}: Traffic analysis and classification tool for open {RAN},'' in \emph{IEEE Int. Conf. Commun.}, 2024, pp. 4894--4899.

\bibitem{Petar2023ORAN}
A.~Casparsen, B.~Soret, J.~J. Nielsen, and P.~Popovski, ``Near real-time data-driven control of virtual reality traffic in open radio access network,'' in \emph{IEEE Glob. Commun. Conf.}, 2023, pp. 3481--3486.

\bibitem{Hoffman2024third_party}
M.~Hoffmann \emph{et~al.}, ``Open {RAN} x{A}pps design and evaluation: {L}essons learnt and identified challenges,'' \emph{IEEE J. Sel. Areas Commun.}, vol.~42, no.~2, pp. 473--486, 2024.

\bibitem{santos2024third_party}
J.~F. Santos, A.~Huff, D.~Campos, K.~V. Cardoso, C.~B. Both, and L.~A. DaSilva, ``Managing {O-RAN} networks: x{A}pp development from zero to hero,'' \emph{arXiv:2407.09619}, 2024.

\bibitem{Fertonani2009Inf}
D.~Fertonani and G.~Colavolpe, ``On reliable communications over channels impaired by bursty impulse noise,'' \emph{IEEE Trans. Commun.}, vol.~57, no.~7, pp. 2024--2030, July 2009.

\bibitem{guruswami2004list}
V.~Guruswami, \emph{List decoding of error-correcting codes: {W}inning thesis of the 2002 ACM doctoral dissertation competition}.\hskip 1em plus 0.5em minus 0.4em\relax Springer Science \& Business Media, 2004, vol. 3282.

\bibitem{Viterbo2021ListViterbi}
M.~Rowshan and E.~Viterbo, ``List viterbi decoding of pac codes,'' \emph{IEEE Trans. Veh. Technol.}, vol.~70, no.~3, pp. 2428--2435, Mar. 2021.

\bibitem{NING2015Spectrum}
Z.~Ning, Y.~Yu, Q.~Song, Y.~Peng, and B.~Zhang, ``Interference-aware spectrum sensing mechanisms in cognitive radio networks,'' \emph{Comput. Electr. Eng.}, vol.~42, pp. 193--206, 2015.

\bibitem{NEURIPS2019Conformalized}
Y.~Romano, E.~Patterson, and E.~Candes, ``Conformalized quantile regression,'' in \emph{Adv. Neural Inf. Process. Syst.}, vol.~32, 2019.

\bibitem{barber2023exc}
R.~F. Barber, E.~J. Candes, A.~Ramdas, and R.~J. Tibshirani, ``Conformal prediction beyond exchangeability,'' \emph{The Annals of Statistics}, vol.~51, no.~2, pp. 816--845, 2023.

\bibitem{bhattacharyya2024groupWCP}
A.~Bhattacharyya and R.~F. Barber, ``Group-weighted conformal prediction,'' \emph{arXiv:2401.17452}, 2024.

\bibitem{ramdas2024majorityvote}
A.~Ramdas and R.~Wang, ``Hypothesis testing with e-values,'' \emph{arXiv:2410.23614}, 2024.

\bibitem{zaheer2018deepsets}
M.~Zaheer, S.~Kottur, S.~Ravanbakhsh, B.~Poczos, R.~R. Salakhutdinov, and A.~J. Smola, ``Deep sets,'' \emph{Proc. Adv. Neural Inf. Process. Syst. (NIPs)}, vol.~30, 2017.

\bibitem{Kingma2014adam}
D.~P. Kingma and J.~Ba, ``Adam: A method for stochastic optimization,'' \emph{arXiv:1412.6980}, 2014.

\bibitem{huang2008cossimilar}
A.~Huang \emph{et~al.}, ``Similarity measures for text document clustering,'' in \emph{Proc. 6th New Zealand Comput. Sci. Res. Student Conf.}, vol.~4, Christchurch, New Zealand, 2008, pp. 9--56.

\bibitem{krzywinski2014boxplot}
M.~Krzywinski and N.~Altman, ``Visualizing samples with box plots: {U}se box plots to illustrate the spread and differences of samples,'' \emph{Nature Methods}, vol.~11, no.~2, pp. 119--121, 2014.

\bibitem{Chowdhury2024Megatron}
M.~Belgiovine, J.~Gu, J.~Groen, M.~Sirera, U.~Demir, and K.~Chowdhury, ``{MEGATRON}: Machine learning in 5{G} with analysis of traffic in open radio access networks,'' in \emph{Proc. Int. Conf. Comput., Netw. Commun. (ICNC),}, 2024, pp. 1054--1058.

\bibitem{Petar2018Traffic}
P.~Popovski, K.~F. Trillingsgaard, O.~Simeone, and G.~Durisi, ``5{G} wireless network slicing for e{MBB}, {URLLC}, and m{MTC}: A communication-theoretic view,'' \emph{IEEE Access}, vol.~6, pp. 55\,765--55\,779, 2018.

\bibitem{Ashist2017Attention}
A.~Vaswani, N.~Shazeer, N.~Parmar, J.~Uszkoreit, L.~Jones, A.~N. Gomez, L.~Kaiser, and I.~Polosukhin, ``Attention is all you need,'' \emph{arxiv:1706.03762}, 2017.

\bibitem{alvarovalcarce2020Nokia}
A.~Valcarce, ``Wireless {S}uite: A collection of problems in wireless telecommunications,'' \url{https://github.com/nokia/wireless-suite}, 2020.

\bibitem{Ana2020RLScheduler}
P.~M. de~Sant~Ana and N.~Marchenko, ``Radio access scheduling using {CMA-ES} for optimized {QoS} in wireless networks,'' in \emph{Proc. IEEE Globecom Workshops}, 2020, pp. 1--6.

\bibitem{tse2005fundamentals}
D.~Tse and P.~Viswanath, \emph{Fundamentals of wireless communication}.\hskip 1em plus 0.5em minus 0.4em\relax Cambridge, U.K.: Cambridge university press, 2005.

\bibitem{heydon2012bluetooth}
R.~Heydon and N.~Hunn, \emph{, Bluetooth Low Energy: The Developer’s Handbook}.\hskip 1em plus 0.5em minus 0.4em\relax Boca Raton, FL, USA: Prentice-Hall, 2012.

\bibitem{Morais2023bluetooth}
D.~H. Morais, \emph{5G NR, Wi-Fi 6, and Bluetooth LE 5 Introduction}.\hskip 1em plus 0.5em minus 0.4em\relax Cham: Springer Nature Switzerland, 2023, pp. 1--7.

\bibitem{correia2024cp_listdecoding}
A.~H.~C. Correia, F.~V. Massoli, C.~Louizos, and A.~Behboodi, ``An information theoretic perspective on conformal prediction,'' \emph{arxiv:2405.02140}, 2024.

\end{thebibliography}

\end{document}